\documentclass{article}

\usepackage{times}
\usepackage{graphicx} 
\usepackage{subfigure} 
\usepackage{siunitx}

\usepackage{natbib}

\usepackage{algorithm}
\usepackage[noend]{algorithmic}
\usepackage{hyperref}

\usepackage{amsmath}
\usepackage{amsthm}
\usepackage{verbatim}
\usepackage{graphicx}
\usepackage{subfigure}
\usepackage{tikz}
\usetikzlibrary{arrows}
\newtheorem{theorem}{Theorem}

\newtheorem{corollary}{Corollary}
\newtheorem{proposition}{Proposition}
\theoremstyle{definition}

\newtheorem{remark}{Remark}
\newtheorem{deff}{Definition}
\newtheorem{claim}{Claim}

\newtheorem{prop}{Property}

\newtheorem{lemma}{Lemma}
\graphicspath{{figures/}}
\usepackage{amsmath}
\usepackage{amsthm}
\usepackage{mathrsfs}
\usepackage{amssymb}
\usepackage{booktabs}
\usepackage{changepage}
\usepackage{hyperref}

\newcommand{\ie}{\textit{i.e.} }

\newcommand{\ind}[1]{ \mathbf{I} \left( #1 \right) }
\newcommand{\nats}[0]{ \mathbb{N} }
\newcommand{\lattice}[0]{ \mathbb{N}^S }
\newcommand{\funcs}[0]{ \mathcal F }
\newcommand{\epsi}[0]{ \varepsilon }
\newcommand{\reals}[0]{ \mathbb{R}^+ }

\renewcommand{\vec}[1]{ \mathbf{ #1 } }
\newcommand{\lone}[1]{ \lVert #1 \rVert_1}

\DeclareMathOperator*{\argmax}{arg\,max}

\begin{document} 

\title{Fast Maximization of Non-Submodular, Monotonic Functions on the Integer Lattice}
\author{Alan Kuhnle, J. David Smith, Victoria G. Crawford, My T. Thai}
\maketitle





\begin{abstract} 
  The optimization of submodular functions on the integer lattice has received
  much attention recently, but the objective functions of many applications are
  non-submodular. 
  We provide two approximation algorithms for maximizing a non-submodular function
  on the integer lattice subject to a cardinality constraint; these are the 
  first algorithms for this purpose that have polynomial query complexity. We propose
  a general framework for influence maximization on the integer lattice that generalizes
  prior works on this topic, 
  and we demonstrate the efficiency of our algorithms in this context.
\end{abstract} 

\section{Introduction}
As a natural extension of finite sets $S$ (equivalently, $\{0,1\}^S$), 
optimization of discrete functions on the integer lattice $\lattice$ has received attention recently
\citep{Alon2012,Demaine2014,Soma2015}. As an example, consider the placement
of sensors in a water network \citep{Krause2008a}; in the set
version, each sensor takes a value in $\{ 0,1 \}$, which corresponds to whether
the sensor was placed. In the lattice version \cite{Soma2015},
each sensor has a power level in
$\{0,\ldots,b\} \subseteq \nats$, to which the sensitivity of the sensor
is correlated. As a second example, consider
the influence maximization problem \citep{Kempe2003}; instead 
of the binary seeding of a user, the lattice version
enables partial incentives or discounts to be used \citep{Demaine2014}.

Although many results from the optimization of submodular set functions have
been generalized to the integer lattice \citep{Soma2015,Soma2015a,Ene2016},
many objective functions arising from applications are not submodular
\citep{Bian2017,Lin2017,Das2011,Horel2016}. In this work, we consider maximization subject to
a cardinality constraint (MCC), where the function $f$ to be maximized may be non-submodular.
Let $k \in \nats$ (the budget), $\vec b \in (\nats \cup \{ \infty \})^S$ (the box), and
let $f: \{ \vec x \in \lattice : \vec x \le \vec b \} \to \reals$ (the objective) be a non-negative
and monotonic\footnote{for all $\vec v \le \vec w$ (coordinate-wise), $f( \vec v) \le f( \vec w)$}
function with $f(\vec 0) = 0$. Then determine
\begin{equation} \label{prob:max} \tag{MCC}
\max_{ \lone{ \vec w} \le k } f( \vec w ),
\end{equation}
where $\vec w = (\vec w_s)_{s \in S} \in \lattice$, $\lone{ \vec w } = \sum_{s \in S} |\vec w_s|$. 

Since the integer lattice may be represented as a multiset of size $k|S|$,
one may use results for Problem \ref{prob:max} with non-submodular set functions.
In particular,
the tight ratio $\frac{1}{\alpha}\left( 1 - e^{-\alpha\gamma_{s}} \right)$ 
of the standard greedy algorithm 
by \citet{Bian2017}, where $\alpha, \gamma_s$ are discussed below, applies with the lattice
adaptation of the standard greedy algorithm (StandardGreedy) given in Alg. \ref{alg:standard-greedy}.  
However, this approach requires $\Omega (|S|k)$ queries of $f$,
which is not polynomial in the input\footnote{The input is considered to be the vector $\vec b$ of length $n = |S|$ and the number $k$ represented in $\log k$ bits (w.l.o.g. each component of $\vec b$ is at most $k$); the function is regarded as an oracle and hence does not contribute to input size.} size $O(|S| \log k)$. Even for applications with 
set functions, $\Omega (|S|k)$ queries may be prohibitive, and researchers 
\citep{Leskovec2007,Mirzasoleiman2014,Badanidiyuru2014} have sought
ways to speed up the StandardGreedy algorithm. Unfortunately, these
approaches rely upon the submodularity of $f$, and there has been no
analogous effort for non-submodular functions.

To quantify the non-submodularity of 
a lattice function $f$, we generalize the following quantities defined for set functions
to the lattice: 
(1) the diminishing-return (DR) ratio $\gamma_d$ of $f$ \citep{Lehmann2006},
(2) the submodularity ratio $\gamma_{s}$ of $f$ \citep{Das2011}, and
(3) the generalized curvature $\alpha$ of $f$ \citep{Bian2017}.
Our main contributions are:
\begin{itemize}
  \item 
    To speed up StandardGreedy (Alg. \ref{alg:standard-greedy}), 
    we adapt the threshold greedy framework
    of \citet{Badanidiyuru2014}
    to non-submodular functions; this 
    yields an algorithm (ThresholdGreedy, Alg. \ref{alg:thresh-greedy}) 
    with approximation ratio $(1 - e^{-\gamma_d\gamma_{s}} - \eta )$,
    for any $\eta > 0$, the first  
    approximation algorithm with polynomial query complexity for Problem \ref{prob:max}
    on the lattice. The query complexity of the StandardGreedy algorithm
    is improved from $\Omega( nk )$ to 
    $O \left( n \log k \log_{\kappa} \left( \epsi^2 /k \right) \right)$,
    where $\kappa, \epsi \in (0,1)$ are parameters of ThresholdGreedy.
  \item We introduce the novel approximation algorithm FastGreedy, which combines elements of
    StandardGreedy and ThresholdGreedy to improve the performance
    ratio to $(1 - e^{-\beta^*\gamma_{s}} - \eta )$, where  
    $\beta^*$ is at least $\gamma_d$ and in many cases\footnote{When the solution $\vec g$ returned by FastGreedy satisfies $\lone{\vec g}=k$. Otherwise, an upper bound on $\beta^*$ is returned.} is determined by the algorithm.
    Furthermore, FastGreedy exploits the non-submodularity of the function
    to decrease its runtime in practice without sacrificing
    its performance guarantee, while maintaining the same worst-case
    query complexity as ThresholdGreedy up to a constant factor.
  \item To demonstrate our algorithms, 
    we introduce a general budget allocation problem for
    viral marketing, which unifies 
    submodular \emph{influence maximization} (IM) 
    under the independent cascade model \citep{Kempe2003}
    with the
    non-submodular \emph{boosting problem} \citep{Lin2017} and
    in addition allows partial incentives.
    We prove a lower bound on the DR and submodularity ratios for this unified
    framework, and we experimentally validate
    our proposed algorithms in this setting.
    
\end{itemize}
\begin{algorithm}[tb]
   \caption{StandardGreedy}
   \label{alg:standard-greedy}
   \begin{algorithmic}[1]
     \STATE {\bfseries Input:} $f \in \funcs_{\vec b}$, $k \in \nats$, $\vec b \in \lattice$
     \STATE {\bfseries Output:} $\vec g \in \lattice$
     \STATE $\vec g \gets \vec 0$
     \FOR{$i = 1$ to $k$}
     \STATE $\vec g \gets \vec g + \argmax_{s \in S : \vec g + \vec s \le \vec b} \delta_{\vec s} ( \vec g )$
     \ENDFOR
     \STATE \textbf{return} $\vec g$
\end{algorithmic}
\end{algorithm}
\section{Related Work} \label{sect:rw}
The study of optimization of submodular set functions is too extensive to give 
a comprehensive overview. On the integer lattice,
there have been many efforts to maximize submodular functions,
\textit{e.g} \citet{Soma2017,Bian2017a,Gottschalk2015}.
To the best of our knowledge, we are the first to study
the optimization of non-submodular functions on the integer lattice. In the following
discussion, we primarily restrict our attention to the maximization of monotonic, submodular
lattice functions
subject to a cardinality constraint and the maximization of non-submodular set functions.
\paragraph{Reduction of \citet{Ene2016}.}
\citet{Ene2016} have given a polynomial-time
reduction from the lattice to a set that enables unified translation of submodular optimization
strategies to DR-submodular (\emph{i.e.} 
DR ratio $\gamma_d = 1$, see Section \ref{sect:dr-ratio})
functions on
the integer lattice. Since this translation is designed for DR-submodular
functions, it does not give a polynomial-time algorithm for Problem \ref{prob:max}
when $f$ is non-submodular. 
Specifically, for the case of maximization subject to a cardinality constraint, \citet{Ene2016}
rely upon the threshold greedy algorithm for submodular set functions \citep{Badanidiyuru2014},
which does not work for non-submodular functions without modifications such as the ones
in our paper. 
\paragraph{Threshold Greedy and Lattice Optimization.}
To speed up the StandardGreedy for submodular set functions, 
\citet{Badanidiyuru2014} introduced the threshold greedy framework, which
speeds up the StandardGreedy algorithm for maximizing
submodular set functions under cardinality constraint from
$O( n k )$ function evaluations to $O \left( \frac{n}{\epsi} \log \frac{n}{\epsi} \right)$,
and it maintains the approximation ratio $(1 - 1/e - \epsi)$, for $\epsi > 0$.
\citet{Soma2015a}
adapted the threshold approach for efficiently maximizing DR-submodular 
functions on the integer lattice and provided
$(1-1/e-\epsi)$-approximation algorithms. Other adaptations of the threshold
approach of \citet{Badanidiyuru2014} to
the integer lattice include \cite{Ene2016,Soma2015}. To the best of our knowledge, in this work we
make the first use of the threshold framework for non-submodular functions.

Our ThresholdGreedy algorithm is an adaptation of the algorithm
of \citet{Soma2015a} for DR-submodular maximization to non-submodular functions.
The non-submodularity requires new analysis, in the following specific ways:
(1) during the binary search phase, we cannot guarantee that we find 
the maximum number of copies whose average gain exceeds the threshold $\tau$;
hence, we must settle for any number of copies whose average gain exceeds $\tau$,
while ensuring that the gain of adding one additional copy falls belows $\tau$. 
(2) To prove the performance ratio, we require a combination of the DR 
ratio $\gamma_d$ 
and the submodularity ratio $\gamma_s$. 
(3) The stopping condition (smallest threshold value) 
is different resulting from the non-submodularity; the proof this condition
is sufficient requires another application of the DR ratio.

\paragraph{Optimization of Non-Submodular Set Functions.}
For non-submodular set functions,
the submodularity ratio $\gamma_s$ was introduced by \citet{Das2011};
we generalize $\gamma_s$ to lattice functions in Section \ref{sect:dr-ratio}, 
and we show the DR ratio $\gamma_d \le \gamma_s$.
\citet{Bian2017} introduced 
\emph{generalized curvature} $\alpha$ of a set function, an analogous concept
to the DR ratio as we discuss in Section \ref{sect:dr-ratio}.
\citet{Bian2017} extended the analysis of \citet{Conforti1984} to non-submodular set functions; together
with the submodularity ratio 
$\gamma_{s}$,
they proved StandardGreedy has tight approximation ratio $\frac{1}{\alpha}\left( 1 - e^{-\gamma_s \alpha} \right)$ under cardinality constraint.

The DR ratio $\gamma_d$ was introduced by \citet{Lehmann2006} 
for the valuation functions in the \emph{maximum allocation problem}
for a combinatorial auction; if each valuation function has DR ratio
at least $\gamma_d$, the maximum allocation problem is a special 
case of maximizing a set function $f$ with DR ratio $\gamma_d$ 
over a matroid, for which \citet{Lehmann2006} employ the StandardGreedy algorithm (for matroids). 

Many other notions of non-submodular set functions have been introduced \citep{Krause2008,Horel2016,Borodin2014,Feige2013}. For a comprehensive discussion of the relation of these and additional
notions to the submodularity ratio $\gamma_s$, we refer the reader to \citet{Bian2017}.
\section{Non-Submodularity on the Lattice} \label{sect:dr-ratio}
In this section, we define the lattice versions of
DR ratio $\gamma_d$, submodularity ratio $\gamma_s$, and
generalized curvature $\alpha$, which are used
in the approximation ratios proved in Section \ref{sect:algs}.
\paragraph{Notations.}
For each $s \in S$, let $\vec s$ be the unit vector with $1$ in
the coordinate corresponding to $s$, and 0 elsewhere.
We write $\delta_{ \vec w } ( \vec v ) = f( \vec v + \vec w ) - f( \vec v)$ for $\vec v, \vec w \in \lattice$.
Given a box in the integer lattice $\vec b \in \lattice$,
let the set of all non-negative, monotonic lattice functions with $f( \vec 0 ) = 0$, and
domain $\{\vec x \in \lattice : \vec x \leq \vec b\}$ be denoted
$\funcs_{\vec b}$.
It is often useful to think of a vector $\vec v \in \lattice$ as
a multi-set containing $\vec v_s$ copies of $s \in S$, where
$\vec v_s$ is the value of $\vec v$'s coordinate corresponding to $s$.
We use the notation $\{\vec v\}$ to represent the multiset corresponding to
the vector $\vec v$.
Finally, we define $\vec v \lor \vec w$ and $\vec v \land \vec w$ for
$\vec v, \vec w \in \lattice$ to be the vector with the coordinate-wise
maximum and minimum respectively.
Rather than an algorithm taking an explicit description of
the function $f$ as input, we consider the function $f$ as an
oracle and measure the complexity of an algorithm in terms
of the number of oracle calls or queries.

We begin with the related concepts of DR ratio and generalized curvature.
\begin{deff}
  Let $f \in \funcs_{\vec b}$.
  The \textit{diminishing-return (DR) ratio of $f$},
  $\gamma_d(f)$, is the maximum value in $[0,1]$ such that
  for any $s \in S$, and for all $\vec v \le \vec w$ such that
  $\vec w + \vec s \leq \vec b$,
  \[ \gamma_d (f) \delta_{\vec s} ( \vec w ) \le \delta_{\vec s} ( \vec v ).\]
\end{deff}
\begin{deff}
  Let $f \in \funcs_{\vec b}$. The \textit{generalized curvature of $f$},
  $\alpha(f)$,
  is the minimum value in $[0,1]$ such that
  for any $s \in S$, and for all $\vec v \le \vec w$ such that
  $\vec w + \vec s \leq \vec b$,
  \[ \delta_{\vec s} ( \vec w ) \ge (1 - \alpha (f)) \delta_{\vec s} ( \vec v ).\]
\end{deff}
The DR ratio extends the notion of DR-submodularity of \citet{Soma2015},
which is obtained as the special case $\gamma_d = 1$.
Generalized curvature for set functions was introduced in
\citet{Bian2017}. Notice that $\alpha$ results in lower bounds
on the marginal gain of $\vec s$ to a vector $\vec w$,
while $\gamma_d$ results in upper bounds on the same quantity:
\[ (1 - \alpha) \delta_{\vec s} ( \vec v ) \le \delta_{\vec s}( \vec w)
  \le \frac{1}{\gamma_d} \delta_{\vec s}(\vec v), \]
whenever $\vec v \le \vec w$ and the above expressions are defined.
Next, we generalize the submodularity ratio of \citet{Das2011} to the
integer lattice.
\begin{deff}
  Let $f \in \funcs_{\vec b}$.
  The \textit{submodularity ratio of $f$},
  $\gamma_s (f)$, is the maximum value in $[0,1]$ such that
  for all $\vec v, \vec w$, such that $\vec v \le \vec w$, 
  \[ \gamma_s(f)( f(\vec w) - f( \vec v ) )
    \le \sum_{s \in \{ \vec w - \vec v \} } \delta_{\vec s} ( \vec v ). \]
\end{deff}
The next proposition, proved in 
Appendix \ref{apx:dr-ratio}, shows the relationship
between DR ratio and submodularity ratio.
\begin{proposition} \label{prop:sm-ratio}
  For all $f \in \funcs_{\vec b}$,
  $\gamma_d(f) \le \gamma_s(f)$.
\end{proposition}
In the rest of this work, we will parameterize functions by the non-submodularity
ratios defined above and partition functions into the sets 
$\funcs_{\vec b}^{\gamma_d,\gamma_s,\alpha} = \{ f \in \funcs_{\vec b} : \gamma_d( f ) = \gamma_d, \gamma_s( f ) = \gamma_s, \alpha(f) = \alpha\}$.
\paragraph{Greedy versions.}
In the proofs of this paper, the full power
of the parameters defined above is not required.
It suffices to consider restricted versions,
where the maximization is taken over only those vectors
which appear in the ratio proofs. 
We define these greedy versions in Appendix \ref{apx:nonsm-defs}
and include more discussion in Remark \ref{rem:greedy-versions} of Section
\ref{sect:threshold-greedy}.
\section{Algorithms}\label{sect:algs}
\subsection{The ThresholdGreedy Algorithm} \label{sect:threshold-greedy}
In this section, we present the algorithm ThresholdGreedy (Alg. \ref{alg:thresh-greedy})
to approximate Problem \ref{prob:max} with ratio $1 - e^{-\gamma_g\gamma_s} - \eta$
with polynomial query complexity. Appendix \ref{apx:bspivot} contains the proofs of all lemmas, claims, 
and omitted details from this section.
\paragraph{Description.} 
ThresholdGreedy operates by considering decreasing thresholds for the marginal gain
in its outer \textbf{for} loop; for each threshold $\tau$, the algorithm
adds on line \ref{line:update-g} elements whose marginal gain exceeds $\tau$ as described below.
The parameter $\kappa \in (0,1)$ determines the stepsize between successive thresholds; the
algorithm continues until the budget $k$ is met (line \ref{line:budget-check}) or the threshold
is below a minimum value dependent on the parameter $\epsi \in (0,1)$.
 
Intuitively, the goal of the threshold approach 
\citep{Badanidiyuru2014} for submodular set functions is as follows.
At each threshold (\textit{i.e.}, iteration of the
outer \textbf{for} loop), add all elements 
whose marginal gain exceeds $\tau$ to the solution $\vec g$. 
On the lattice, adding all copies 
of $s \in S$ whose average gain exceeds $\tau$ on line \ref{line:update-g}
would require the addition of the maximum multiple $l\vec s$ such
that the average marginal gain exceeds $\tau$: 
\begin{equation} \label{prop:1} \tag{P1}
\delta_{l\vec s} (\vec g) \ge l \tau,
\end{equation}
as in the threshold algorithm of \citet{Soma2015a} for DR-submodular functions,
in which the maximum $l$ is identified by binary search.
However, since $f$ is not DR-submodular, it is not always the case that
$\delta_{\vec s}( \vec g + l \vec s) \ge \delta_{\vec s}( \vec g + (l + 1) \vec s )$,
for each $l$. 
For this reason, we cannot find the maximum such $l$ by binary search. 
Furthermore, even if we found the maximum $l$ for each $s \in S$,
we could not guarantee that all elements of marginal gain at least $\tau$
were added due to the non-submodularity of $f$: an element whose gain is
less than $\tau$ when considered in the inner \textbf{for} loop might have
gain greater than $\tau$ after additional elements are added to the solution.

ThresholdGreedy 
more conservatively ensures that the number $l$
chosen for each $s \in S$ satisfies both (\ref{prop:1}) and
\begin{equation} \tag{P2}
  \delta_{\vec s} ( \vec g + l\vec s ) < \tau, \label{prop:2}
\end{equation}
but it is not necessarily the maximum such $l$.
\paragraph{Pivot.}
Any $l$ satisfying both (\ref{prop:1}) and (\ref{prop:2}) is termed
a \emph{pivot}\footnote{For convenience, we also define the maximum value of $l$,
$l_{max} = \min \{ \vec b_s - \vec g_s, k - \lone{\vec g} \}$
to be a pivot if $l_{max}$ satisfies (\ref{prop:1}) only,
and set $\delta_{\vec s}(\vec g + l_{max} \vec s) = 0$, so
that all pivots satisfy both properties.} 
with respect to $\vec g, s, \tau$.
Perhaps surprisingly, a valid pivot can be found 
with binary search in $O(\log \vec b_{max} ) = O(\log k)$ function queries, where
$\vec b_{max} = \max_{s \in S} \vec b_{s}$; 
discussion of BinarySearchPivot and proof of this results is provided in Appendix \ref{apx:bspivot},
Lemma \ref{lemm:pivot}.
By finding a pivot for each $s \in S$, 
ThresholdGreedy does not attempt to add all elements exceeding the marginal
gain of threshold $\tau$; instead, 
ThresholdGreedy maintains the following property 
at each threshold.
\begin{prop} \label{lemm:exist}
  Let $\vec g_{\tau}$ be the solution
  of ThresholdGreedy immediately after the iteration
  of the outer \textbf{for} loop corresponding to
  threshold $\tau$. Then for each $s \in S$, 
  there exists $\vec h \le \vec g_{\tau}$ such that
  $\delta_{\vec s} ( \vec h ) < \tau$.
\end{prop}
\begin{algorithm}[tb]
   \caption{ThresholdGreedy}
   \label{alg:thresh-greedy}
\begin{algorithmic}[1]
  \STATE {\bfseries Input:} $f \in \funcs_{\vec b}$, $k \in \nats$, $\kappa, \epsi \in (0,1)$.
  \STATE {\bfseries Output:} $\vec g \in \lattice$
  \STATE $\vec g \gets \vec 0$, $M \gets \max_{s \in S} f( \vec s )$.
  \FOR{$\left( \tau = M; \tau \ge \frac{\kappa \epsi^2 M }{k}; \tau \gets \kappa \tau \right)$\label{fg-outer-for}}
  \FOR{$s \in S$}
  \STATE $l \gets $BinarySearchPivot$( f, \vec g, \vec b, s, k, \tau )$
  \STATE $\vec g \gets \vec g + l \vec s$\label{line:update-g}
  \IF{$\lone{\vec g} = k$\label{line:budget-check}}
  \STATE \textbf{return} $\vec g$
  \ENDIF
  \ENDFOR
  \ENDFOR
  \STATE \textbf{return} $\vec g$
\end{algorithmic}
\end{algorithm}
\paragraph{Performance ratios.} Next, we present the main result of this
section, the performance guarantee involving the DR and submodularity ratios.
 Observe that the query complexity of ThresholdGreedy is polynomial
in the input size $O( n \log k )$. 
\begin{theorem} \label{thm:threshold}
  Let an instance of Problem \ref{prob:max} be given,
  with $f \in \funcs_{\vec b}^{\gamma_d,\gamma_s,\alpha}$. 
  If $\vec g$ is the solution returned by ThresholdGreedy
  and $\vec \Omega$ is
  an optimal solution to this instance, then
  \begin{align*} 
    f( \vec g ) &\ge \left(1 - e^{- \kappa \gamma_d \gamma_s }- \epsi \right) f( \vec \Omega ).
  \end{align*}
  The query complexity of ThresholdGreedy is
  $O \left( n \log k \log_{\kappa} \left( \epsi^2 /k \right) \right)$.
\end{theorem}
If $\eta > 0$ is given, the assignment $\kappa = (1 - \eta / 2)$, $\epsi = \eta / 2$
yields performance ratio at least $1 - e^{-\gamma_d \gamma_s} - \eta$.
\begin{proof}
  If $\gamma_d < \epsi$, the ratio holds trivially; 
  so assume $\gamma_d \ge \epsi$.
  The proof of the following claim requires an application of 
  the DR ratio.
  \begin{claim} \label{lemm:wlog}
    Let $\vec g$ be produced by a modified version of ThresholdGreedy
    that continues until $\lone{ \vec g} = k$.
    If we show 
    $f( \vec g ) \ge (1 - e^{-\kappa \gamma_d \gamma_s})f( \vec \Omega )$, the
    results follows.
  \end{claim}
  Thus, for the rest of the proof let $\vec g$ be as described in 
  Claim \ref{lemm:wlog}.
  Let $\vec g^{t}$ be the value of $\vec g$ after the $t$th execution of line
  \ref{line:update-g} of ThresholdGreedy.
  Let $l^t$ be the $t$th pivot, such that $\vec g^t = \vec g^{t - 1} + l^t\vec s^t$.
  The next claim lower bounds the marginal gain in terms of the DR ratio
  and the previous threshold.
  \begin{claim} \label{lemm:marge}
    For each $s \in \{ \vec \Omega - \vec g^{t-1} \wedge \vec \Omega \}$,
    $$l^t \gamma_d \kappa \delta_{\vec s} ( \vec g^{t - 1} ) \le f( \vec g^t ) - f( \vec g^{t - 1}).$$
  \end{claim} 
    \begin{proof}
      Let $\tau$ be the threshold at which
    $l^t\vec s^t$ is added to $\vec g^{t-1}$;
    let $s \in \{ \vec \Omega - \vec g^{t - 1} \wedge \vec \Omega \}$.
    If $\tau$ is the first threshold, 
    $$\gamma_d \delta_{\vec s} ( \vec g^{t - 1} ) \le \delta_{\vec s}( \vec 0 ) \le \tau < \frac{\tau}{\kappa }.$$
    If $\tau$ is not the first threshold, $\tau' = \tau / \kappa$
    is the previous threshold value of the previous iteration of the
    outer \textbf{for} loop. 
    By Property \ref{lemm:exist},
    there exists $\vec h \le \vec g_{\tau'} \le \vec g^{t - 1}$, such that 
    $\delta_{\vec s}( \vec h ) < \tau'$.
    By the definition of DR ratio,
    $\gamma_d \delta_{\vec s} ( \vec g^{t-1} ) \le \delta_{\vec s}( \vec h ) < \tau' = \tau / \kappa$.
    
    In either case, by the fact that property (\ref{prop:1}) of a pivot
    holds for $l^t$, we have
    $$f( \vec g^t ) - f( \vec g^{t-1} ) \ge l^t\tau \ge l^t \gamma_d \kappa \delta_{\vec s}( \vec g^{t-1} ).\qedhere$$
  \end{proof}
  Since $| \vec \Omega | \le k$, we have by Claim \ref{lemm:marge}
  \begin{align*}
    f( \vec g^t ) - f( \vec g^{t-1} ) &\ge \frac{l^t \gamma_d \kappa}{k} \sum_{s \in \left\{ \vec \Omega - (\vec g^{t-1} \wedge \vec \Omega ) \right\}} \delta_{\vec s} ( \vec g^{t - 1} ) \\
                                      &= \frac{l^t \gamma_d \kappa}{k} \sum_{s \in \left\{ \vec \Omega \vee \vec g^{t-1} - \vec g^{t - 1} \right\}} \delta_{\vec s} ( \vec g^{t - 1} ) \\
                                      &\ge \frac{ l^t \gamma_d \gamma_s \kappa }{k} \left( f ( \vec \Omega \vee \vec g^{t-1}) - f( \vec g^{t-1} ) \right)\\
    &\ge \frac{ l^t \gamma_d \gamma_s \kappa }{k} \left( f ( \vec \Omega ) - f( \vec g^{t-1} ) \right),
  \end{align*}
  where the equality follows from
  the lattice identity $\vec v \vee \vec w - \vec v = \vec w - \vec v \wedge \vec w$
  for all $\vec v, \vec w \in \lattice$, the
  second inequality is by definition of the submodularity ratio, and the third
  inequality is from monotonicity. 
  From here, we obtain 
  $f( \vec g ) \ge \left( 1 - \prod_{t=1}^T \left( 1 - \frac{l^t \gamma_d \gamma_s \kappa}{k} \right) \right) f( \vec \Omega ),$
  from which the hypothesis of Claim \ref{lemm:wlog} follows.
  \paragraph{Query complexity.} The \textbf{for} loop on line \ref{fg-outer-for}
  (Alg. \ref{alg:thresh-greedy}) 
  iterates at most $\log_{\kappa} \epsi^2 / k$ times;
  each iteration requires $O( n \log k )$ queries, by 
  Lemma \ref{lemm:pivot}.
\end{proof}
For additional speedup, the inner
\textbf{for} loop of FastGreedy may be parallelized,
which divides the factor of $n$ in the query complexity by the number of threads but worsens
the performance ratio; in addition to $\gamma_d,\gamma_s$, 
the generalized curvature $\alpha$ is required in the proof.
\begin{corollary} \label{cor:thresh-par}
  If the inner \textbf{for} loop of ThresholdGreedy is parallelized, 
  the performance ratio becomes $1 - e^{-(1 - \alpha)\gamma_d\gamma_s} - \eta$,
  for $\eta > 0$.
\end{corollary}
\begin{remark} \label{rem:greedy-versions}
  A careful analysis of the usage of $\gamma_d$,$\gamma_s$ in the proof of
  Theorem \ref{thm:threshold} shows that the full power of the definitions
  of these quantities 
  is not required. Rather, it is sufficient to consider ThresholdGreedy versions
  of these parameters, as defined in Appendix \ref{apx:nonsm-defs}. In the same
  way, we also have FastGreedy version of $\gamma_s$ based upon the proof
  of Theorem \ref{thm:fg-ratio}. The FastGreedy version of the DR ratio is 
  an integral part of how the algorithm works and is calculated
  directly by the algorithm, as we discuss in the next section.
\end{remark}
\subsection{The FastGreedy Algorithm} \label{sect:fg}
The proof of the performance ratio of ThresholdGreedy requires both 
the submodularty ratio $\gamma_s$ and the DR ratio $\gamma_d$. 
In this section, we provide an algorithm (FastGreedy, Alg. \ref{alg:fast-greedy}) 
that achieves ratio $1 - e^{-\beta^*\gamma_s} - \eta$, with
factor $\beta^* \ge \gamma_d$ that it can determine during its execution.
Appendix \ref{apx:fg} provides proofs for all lemmas, claims, and omitted details.
\paragraph{Description.} 
FastGreedy employs a threshold framework analogous to ThresholdGreedy.
Each iteration of the outer \textbf{while} loop of FastGreedy
is analogous to an iteration of the outer \textbf{for} loop in 
ThresholdGreedy, in which elements are added whose marginal
gain exceeds a threshold. 
FastGreedy employs BinarySearchPivot to
find pivots for each $s \in S$ for each threshold value $\tau$. 
Finally, the parameter $\epsi$ determines a minimum threshold value.

As its threshold, FastGreedy uses $\tau = \beta \kappa m$, 
where $m$ is the maximum marginal
gain found on line \ref{line:greedy}, 
parameter $\kappa$ is the intended stepsize between thresholds
as in ThresholdGreedy, and $\beta$ is an upper bound 
on the DR ratio $\gamma_d$, as described below.
This choice of $\tau$ has the following advantages over
the approach of ThresholdGreedy: (1) since the threshold is related
to the maximum marginal gain $m$, the theoretical performance ratio 
is improved; 
(2) the use of $\beta$ to lower the threshold 
ensures the same\footnote{Up to a constant factor, which depends
on $\gamma_d$.} worst-case query complexity
as ThresholdGreedy and leads to substantial 
reduction of the number of queries in practice, as we
demonstrate in Section \ref{sect:experiments}.
\paragraph{FastGreedy DR ratio $\beta^*$.}
If FastGreedy is modified\footnote{This modification can be accomplished by setting $\epsi$ to ensure the condition on line \ref{fg-while} is always true on this instance.} to continue until
$\lone{\vec g} = k$, let the final, smallest value $\beta^*$
of $\beta$ be termed the \emph{FastGreedy DR ratio} on
the instance. 
The FastGreedy DR ratio $\beta^*$ is at least the DR ratio $\gamma_d$
of the function, up to the parameter $\delta$:  
\begin{lemma} \label{lemm:beta}
  Let parameters $\kappa, \delta, \epsi \in (0,1)$  be given.
  Throughout the execution of FastGreedy on an instance of Problem \ref{prob:max} with $f \in \funcs^{\gamma_d,\gamma_s}_\vec b$, $\beta \ge \gamma_d\delta$.
  Since $\epsi$ can be arbitrarily small, $\beta^* \ge \gamma_d \delta$.
\end{lemma}
\begin{proof}
  Initally, $\beta = 1$; it decreases by a factor of $\delta \in (0,1)$
  at most once per iteration of the \textbf{while} loop. 
  Suppose $\beta \le \gamma_d$ for some iteration 
  $i$ of the \textbf{while} loop, and let $\vec g$ have the value 
  assigned immediately after iteration $i$, $m$ have the value
  assigned after line \ref{line:greedy} of iteration $i$. 
  Since a valid pivot was found 
  for each $s \in S$ during iteration $i$, by property (\ref{prop:2}) there exists
  $\vec g^s \le \vec g$, $\delta_{\vec s} ( \vec g^s ) < \beta \kappa m \le \gamma_d \kappa m$.
  Hence $\delta_{\vec s} (\vec g) \le \kappa m $, by the definition of DR ratio.
  In iteration $i + 1$, $m'$ has the value of $m$ from iteration $i$,
  so the value of $m$ computed during iteration $i + 1$ is at 
  most $\kappa m'$, and $\beta$ does not decrease during iteration $i + 1$.
\end{proof}
\begin{algorithm}[tb]
   \caption{FastGreedy}
   \label{alg:fast-greedy}
   \begin{algorithmic}[1]
     \STATE {\bfseries Input:} $f \in \funcs_{\vec b}$, $k \in \nats$, $\kappa, \delta, \epsi \in (0,1)$.
     \STATE {\bfseries Output:} $\vec g \in \lattice$
     \STATE $\vec g \gets 0$, $M \gets \max_{s \in S} f(\vec s)$, $m \gets M, m' \gets M / \kappa$, $\beta \gets 1$
     \WHILE{$m \ge M \epsi^2 / k$\label{fg-while}}
     \STATE $m \gets \max_{s \in S} \delta_{\vec s}( \vec g )$\label{line:greedy} 
     \IF{$m > \kappa m'$}\label{line:nextThresh}
     \STATE $\beta \gets \beta \delta$\label{line:reviseBeta}
     \ENDIF
     \STATE $m' \gets m$\label{line:newmp} 
     \STATE $\tau \gets \beta \kappa m$
     \FOR{$s \in S$\label{fg-for}}
     \STATE $l \gets $BinarySearchPivot$( f, \vec g, \vec b, s, k, \tau )$ \label{line:fgfp}
     \STATE $\vec g \gets \vec g + l\vec s$\label{line:greedyAdd}
     \IF{$\lone{\vec g} = k$}
     \STATE \textbf{return} $\vec g$
     \ENDIF
     \ENDFOR
     \ENDWHILE
   \STATE \textbf{return} $\vec g$
\end{algorithmic}
\end{algorithm}
\paragraph{Performance ratio.}
Next, we present the main result of this section. In contrast to
ThresholdGreedy, the factor of $\gamma_d$ in the performance ratio
has been replaced with $\beta^*$; at the termination of the algorithm,
the value of $\beta^*$ is an output of FastGreedy if the solution $\vec g$ satisfies 
$\lone{ \vec g } = k$. In any case, by Lemma \ref{lemm:beta}, the performance ratio
is at worst the same as that of ThresholdGreedy. 
\begin{theorem} \label{thm:fg-ratio}
  Let an instance of Problem \ref{prob:max} be given,
  with $f \in \funcs_{\vec b}^{\gamma_d,\gamma_s}$.
  Let $\vec g$ be the solution returned by FastGreedy with 
  parameters $\kappa, \delta, \epsi \in (0,1)$,
  and let $\vec \Omega$ be
  an optimal solution to this instance; also, suppose $\gamma_d \ge \epsi$. Let $\beta^*$ be the
  FastGreedy DR ratio on this instance. Then,
  \begin{align*} 
    f( \vec g ) &\ge \left(1 - e^{-\kappa \beta^* \gamma_s }- \epsi \right) f(\vec \Omega )
  \end{align*}
  The worst-case query complexity of FastGreedy is
  $O \left( \left( \log_{\delta}(\gamma_d) \log_{\kappa}( \gamma_d ) + \log_{\kappa} \epsi^2 / k \right) n \log k \right).$
\end{theorem}
If $\eta > 0$ is given, the assignment $\kappa = (1 - \eta / 2)$, $\epsi = \eta / 2$
yields performance ratio at least $1 - e^{-\beta^* \gamma_s} - \eta$.
\begin{proof}[Proof of query complexity]
  The performance ratio is proved in Appendix \ref{apx:fg}.
  Let $m'_1, \ldots, m'_K$ be the sequence of $m'$ values in the order
  considered by the algorithm. By Lemma \ref{lemm:beta}, $m'_j > \kappa m'_{j - 1}$ 
  at most 
  $\Gamma = \log_{\delta} \gamma_d$ times; label each such
  index $j$ an \emph{uptick}, and let $j_1, \ldots, j_l$ be the indices
  of each uptick in order of their appearance. Also, let $k_i$ be the
  first index after $j_i$ such that $m'_{k_i} \le \kappa m'_{j_i - 1}$, for
  each $i \in \{1, \ldots, l\}$.

  Next, we will iteratively delete from the sequence of $m'$
  values. Initially, let $\ell = l$ be the last uptick 
  in the sequence; delete all terms $m'_{j_\ell}, \ldots, m'_{k_\ell - 1}$
  from the $m'$ sequence. Set $\ell = \ell - 1$ and repeat this process 
  until $\ell = 0$. 
  \begin{claim} \label{claim:fast}
    For each $\ell$ selected in the iterative
    deletion above, there are 
    at most $\log_{\kappa} \gamma_d$
    values deleted from the sequence.
  \end{claim}
  By Claim \ref{claim:fast} and the bound on the number of upticks, 
  we have deleted at most 
  $\log_{\kappa} \gamma_d \log_{\delta} \gamma_d$ thresholds
  $m'$ from the sequence; every term in the remaining sequence 
  satisfies $m'_j \le \kappa m'_{j-1}$; hence, the remaining sequence
  contains at most $\log_{\kappa} \epsi^2 / k$ terms, by
  its initial and terminal values. The query complexity follows
  from the number of queries per value of $m'$, which is $O( n \log k )$ by
  Lemma \ref{lemm:pivot}. 
\end{proof}
\section{Influence Maximization: A General Framework} \label{sect:IM}
In this section, we provide a non-submodular framework for
viral marketing on a social network that unifies
the classical influence maximization \citep{Kempe2003} 
with the boosting problem \citep{Lin2017}.
\paragraph{Overview.}
The goal of influence maximization 
is to select seed users (\ie initially activated users) to maximize the expected
adoption in the social network, where the total number of seeds is restricted by a budget,
such that the expected adoption in the social network is maximized. 
The boosting problem
is, given a fixed seed set $S$, to incentivize (\ie increase the susceptibility of a user
to the influence of his friends) 
users within a budget such that the expected adoption with seed set $S$
increases the most. 

Our framework combines the above two scenarios with a partial incentive: 
an incentive (say, $x$\% off the purchase price) 
increases the probability a user will purchase the
product independently and increases
the susceptibility of the user to the influence of his friends. Hence, our problem 
asks how to best allocate the budget between (partially) seeding users and boosting the
influence of likely extant seeds. Both
the classical influence 
maximization and the non-submodular boosting problem can be obtained as special cases,
as shown in Appendix \ref{apx:im}.

Our model is related to the formulation of \citet{Demaine2014}; however, they employ
a submodular threshold-based model, while our model is inherently non-submodular 
due to the boosting mechanism \citep{Lin2017}.
Also,
GIM is related to the submodular budgeted allocation problem of
\citet{Alon2012}, in which the influence of an advertiser increases
with the amount of budget allocated; the main difference with GIM is that
we modify incoming edge weights with incentives instead of outgoing, which
creates the boosting mechanism responsible for the non-submodularity.
\paragraph{Model.} Given a social network $G = (V,E)$, and a product
$p$, we define the following model of adoption. The 
allocation of budget to $u$ is
thought of as a discount towards purchasing the product;
this discount increases the probability
that this user will adopt or purchase the product. Furthermore,
this discount increases the susceptibility of the user to influence from its (incoming) social
connections. 

Formally, an incentive level 
$\vec x_u$ is chosen for each user $u$. With independent 
probability $p(u, \vec x_u)$, user $u$ initially activates or adopts
the product; altogether, this creates a probabilistic initial set 
$S$ of activated users. Next, through the
classical Independent Cascade (IC) model\footnote{The IC model is defined in Appendix \ref{apx:im}.} 
of adoption, users influence their
neighbors in the social network; wherein the weight $p(v,u, \vec x_u)$ for edge
$(v,u)$ is determined by the incentive level $\vec x_u$ of user $u$
as well as the strength of the social connection from $v$ to $u$.


We write $p^{\vec x}(H, T)$ to denote the probability of full graph realization $H$ and
seed set $T$ when $\vec x$ gives the incentive levels for each user.
We write $R(H,T)$ to denote the size
of the reachable set from $T$ in realization $H$.
The expected activation in the network 
given a choice $\vec x$ of incentive levels is given by
$\mathbb{I}( \vec x ) = \sum_{T \subseteq V}\sum_{H \subseteq G} p^{\vec x}(H,T) R(H,T),$
where an explicit formula for $p^{\vec x}(H,T)$ is given in Appendix
\ref{apx:im}.  
Finally, let $\mathbb{A}( \vec x ) = \mathbb{I}( \vec x ) - \mathbb{I}( \vec 0 )$.
\begin{deff}[Generalized Influence Maximization (GIM)]
Let social network $G = (V,E)$ be given, together with the mappings
$i \mapsto p(u, i)$, $i \mapsto p(u,v,i)$, for all $u \in V, (u,v) \in E$,
for each $i \in \{0,\ldots,L\}$, where $L$ is the number of incentive levels.
Given budget $k$, 
determine incentive levels $\vec x$, with $\lone{\vec x} \le k$,
such that $\mathbb A ( \vec x )$ is maximized.
\end{deff}
\paragraph{Bound on non-submodularity.} 
Next, we provide a lower bound on the greedy submodularity ratios (see Appendix \ref{apx:nonsm-defs}).
We emphasize that the assumption that the probability mappings as a function of incentive level
be submodular does not imply the objective
$\mathbb{A}(\vec x)$ is DR-submodular.
Theorem \ref{thm:dr-lower-bound} is proved in Appendix \ref{apx:im}.
\begin{theorem} \label{thm:dr-lower-bound}
Let $\mathcal I$ be an instance of GIM, with budget $k$.
Let ${c_e} = \max_{(u,v) \in E, i \in L} \frac{p(u,v, i + 1)}{p(u,v, i)}$, $c_n = \max_{x \in V, i \in L} \frac{p(x, i + 1)}{p(x, i)}$. Suppose for all $(u,v) \in E, w \in V$, the mappings $i \mapsto p(u,v, i)$, $i \mapsto p(w,i)$
are submodular set functions.
Then, the greedy submodularity ratios defined in Appendix \ref{apx:nonsm-defs} and the
FastGreedy DR ratio are lower bounded by $c_e^{-k\Delta} c_n^{-k},$
where $\Delta$ is the maximum in-degree in $G$.
\end{theorem}
\section{Experimental Evaluation} \label{sect:experiments}
In this section, we evaluate our proposed algorithms for the
GIM problem defined in Section \ref{sect:IM}.
We evaluate our algorithms as compared with StandardGreedy;
by the naive reduction of the lattice to sets in exponential time,
this algorithm is equivalent performing this reduction and running 
the standard greedy for sets, the performance of which for non-submodular
set functions was analyzed by \citet{Bian2017}.

In Section \ref{sect:methods}, we describe our methodology;
in Section \ref{sect:experiments-comparison}, we compare
the algorithms and non-submodularity parameters.
In Appendix \ref{sect:experiments-parameters},
we explore the behavior of FastGreedy as the parameters $\delta,\kappa$, and
$\epsi$ are varied.
\subsection{Methodology} \label{sect:methods}
Our implementation uses Monte Carlo sampling to
estimate the objective value $\mathbb{A}( \vec x )$, with \num{10000} samples
used. As a result, each function
query is relatively expensive. 

We evaluate on two networks taken from the SNAP dataset \cite{snapnets}:
\texttt{ca-GrQc} (``GrQc''; 15k nodes, 14.5K edges) and
\texttt{facebook} (``Facebook''; 4k nodes, 176K edges).  Unless
otherwise specified, we use 10 repetitions per datapoint and display the
mean. The width of shaded intervals is one standard deviation.
Standard greedy is omitted from some figures where running time is
prohibitive. Unless noted otherwise, we use default settings of
$\epsi = 0.05$, $\delta = 0.9$, $\kappa = 0.95$. We use a uniform
box constraint and assign each user the same number of incentive levels;
the maximum incentive level for a user corresponds to giving the product
to the user for free and hence deterministically seeds the user; we 
adopt linear models for the mappings $i \mapsto p(u,i)$, $i \mapsto p(u,v,i)$.
We often plot versus $K$, which is defined as the maximum number of deterministic
seeds; for example, if $k = 200$ with $10$ incentive levels, then $K= 20$.

\subsection{Results}\label{sect:experiments-comparison}
\begin{figure}[t]
  \centering
  \subfigure[GrQc (10 levels)]{\includegraphics[width=0.22\textwidth,height=0.13\textheight]{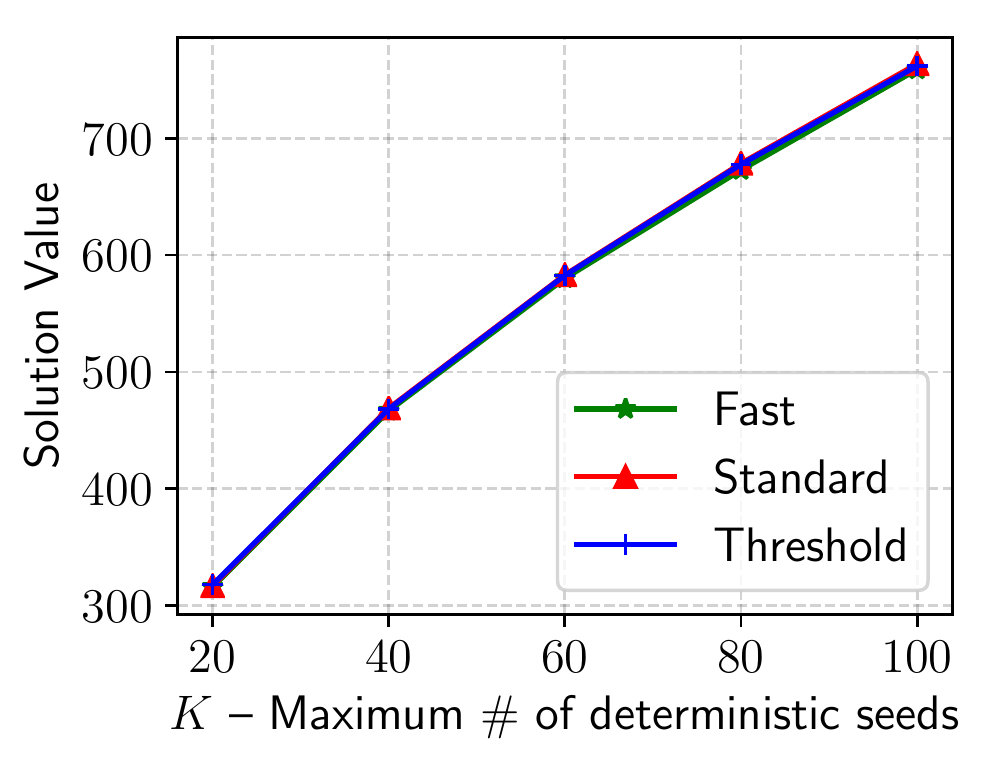}\label{fig:grqc-qos}}
  \subfigure[Facebook (100 levels)]{\includegraphics[width=0.22\textwidth,height=0.13\textheight]{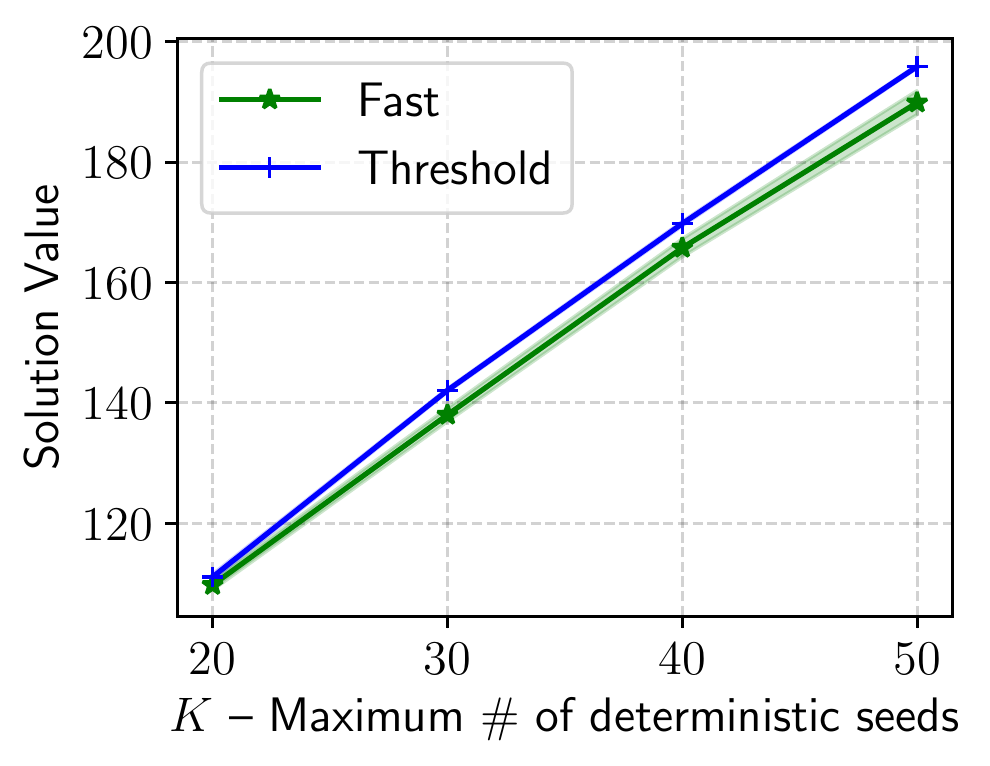}\label{fig:fb-qos}}
  \caption{Activation $\mathbb{A}( \vec g )$ for the solution returned by each algorithm.} \label{fig:qos}
\end{figure}
\begin{figure}[t]
  \centering
    \subfigure[GrQc (10 levels)]{\includegraphics[width=0.22\textwidth,height=0.13\textheight]{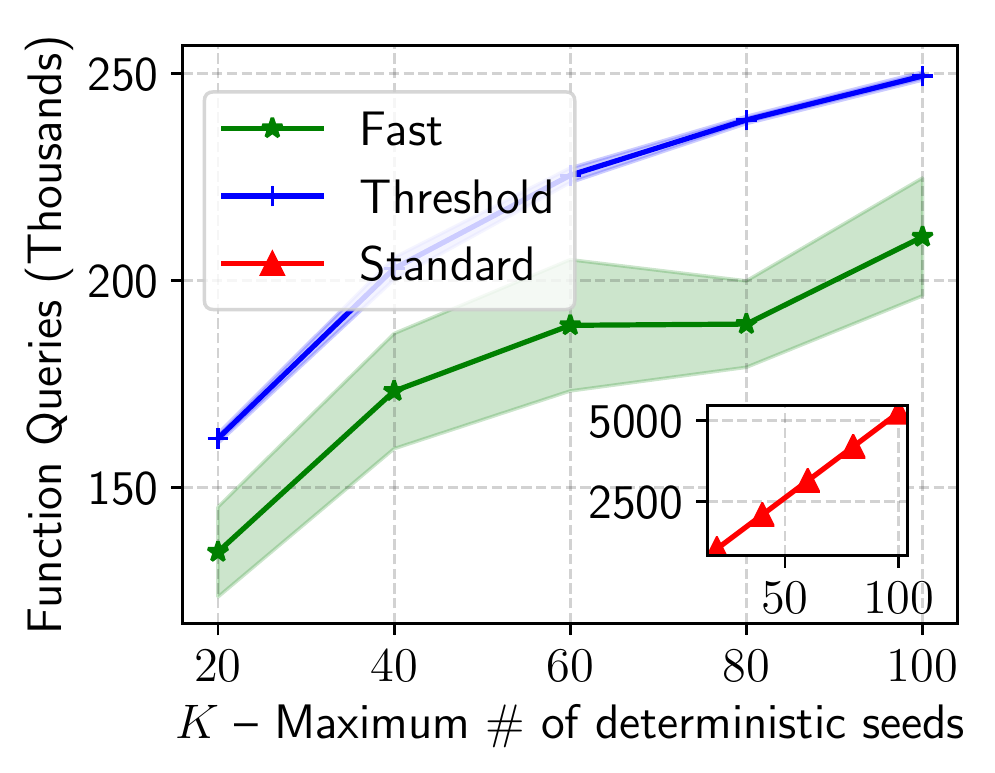}}
    \subfigure[Facebook (100 levels)]{\includegraphics[width=0.22\textwidth,height=0.13\textheight]{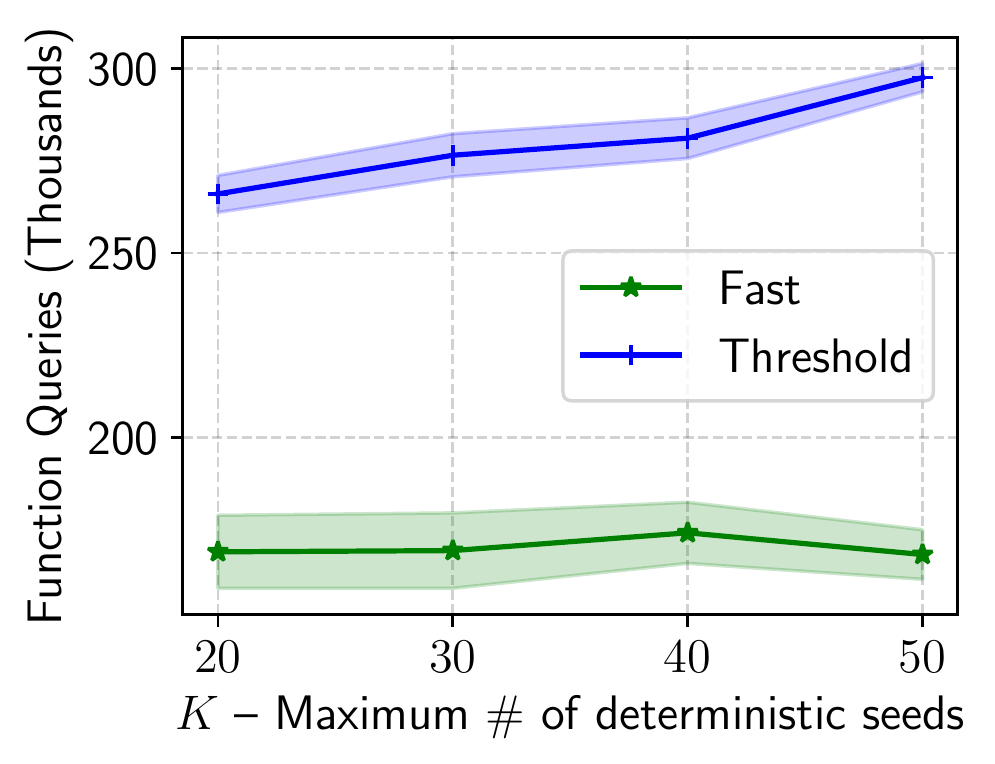}\label{fig:q-fb}}
  \caption{Total function queries on the GrQc and Facebook networks.} \label{fig:queries}
\end{figure}
\begin{figure}[t]
  \centering
  \subfigure[GrQc (10 levels)]{\includegraphics[width=0.22\textwidth,height=0.13\textheight]{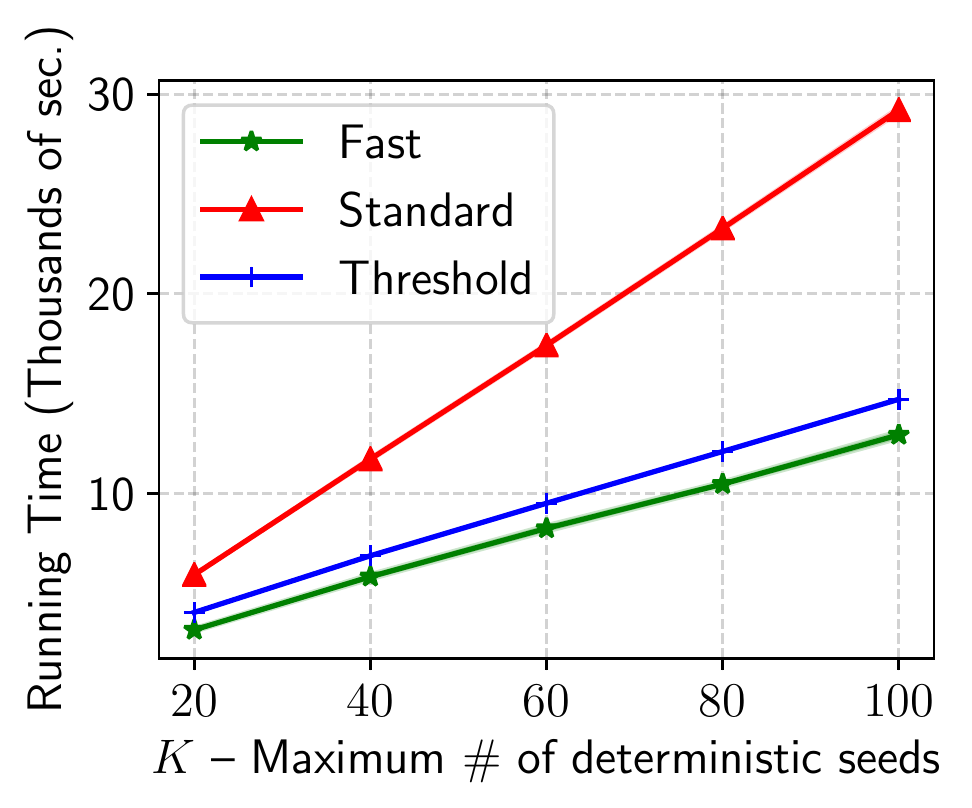}}
  \subfigure[Facebook (100 levels)]{\includegraphics[width=0.22\textwidth,height=0.13\textheight]{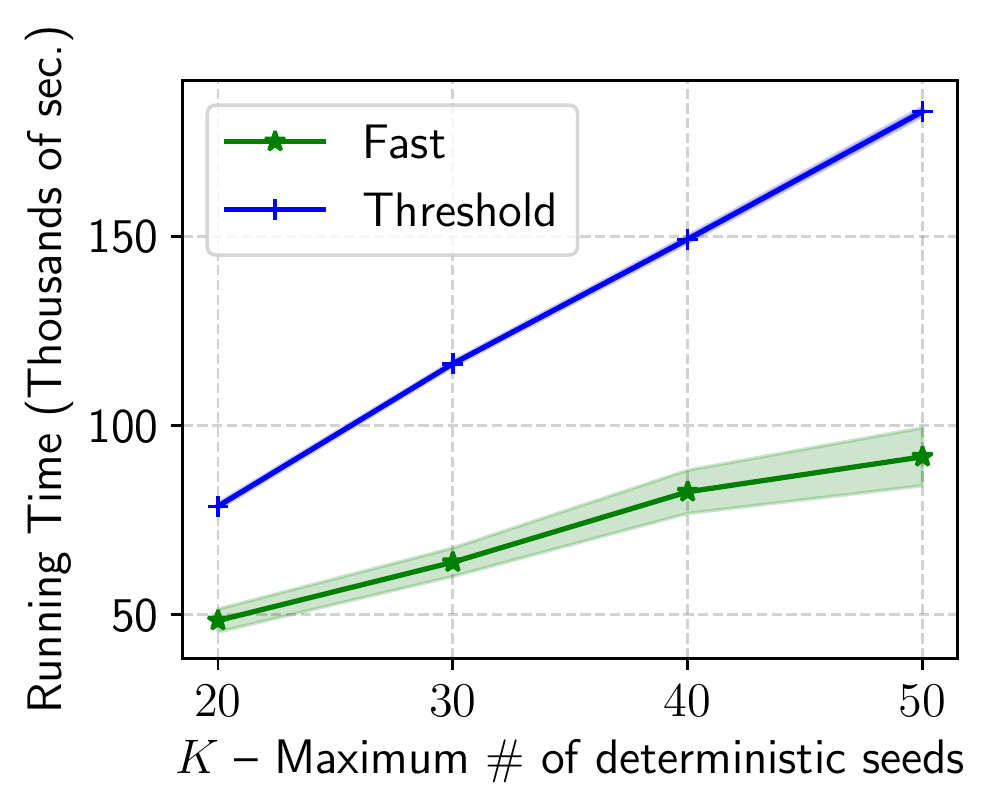}}
  \caption{Runtime on the GrQc and Facebook networks with 100 levels.} \label{fig:runtime}
\end{figure}
\begin{figure}[t] 
  \centering
  \subfigure[GrQc]{\includegraphics[width=0.22\textwidth,height=0.13\textheight]{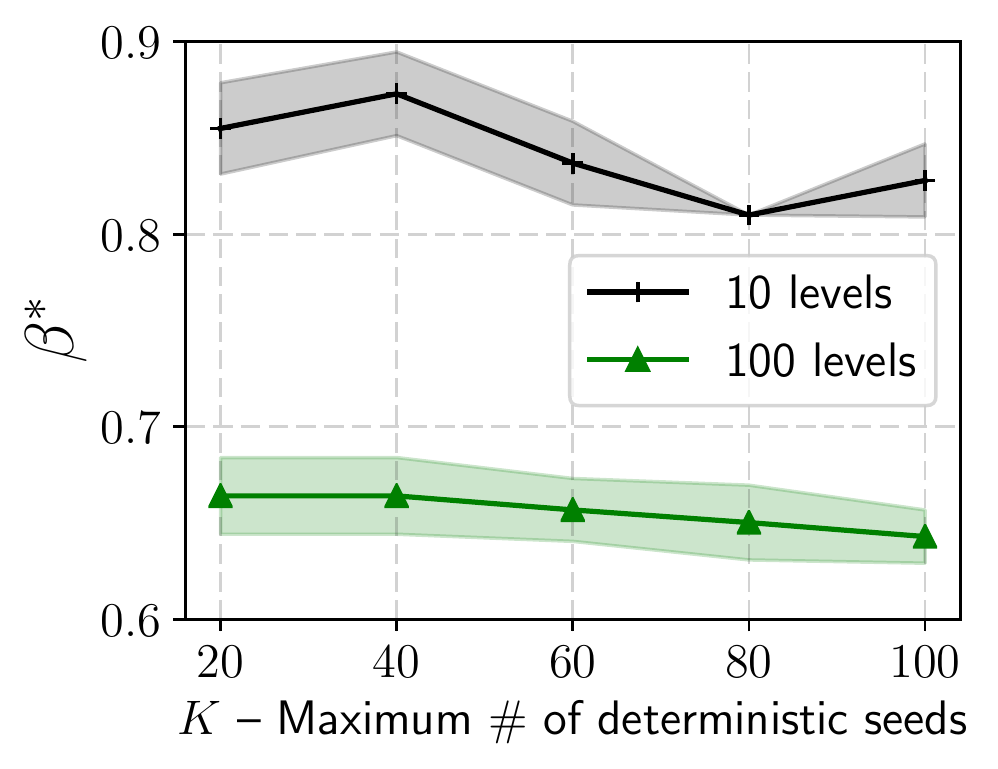}\label{fig:beta}}
  \subfigure[BA network]{\includegraphics[width=0.22\textwidth,height=0.13\textheight]{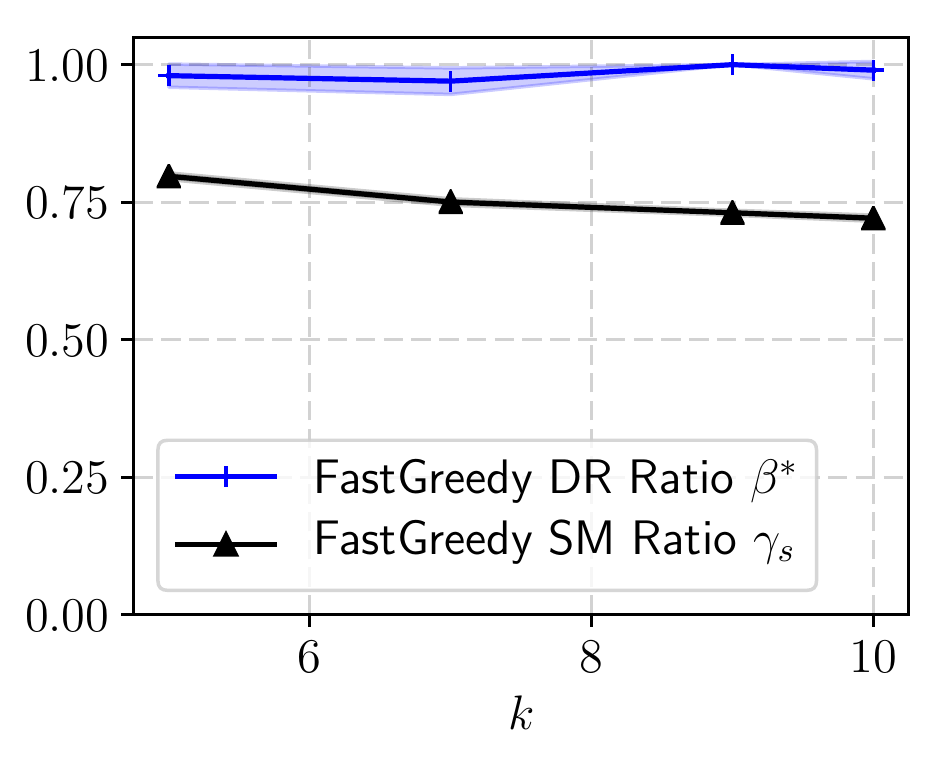}\label{fig:gamma}}
  \caption{\textbf{(a):} The value of FastGreedy DR ratio $\beta^*$ on the GrQc dataset. \textbf{(b):} FastGreedy submodularity ratio $\gamma_s$ and FastGreedy DR Ratio $\beta^*$ on a small, random BA network.} \label{fig:gamma-beta}
\end{figure}
In this section, we demonstrate the following: (1) our algorithms exhibit
virtually identical quality of solution with StandardGreedy,
(2) our algorithms query the function much fewer times, which
leads to dramatic runtime improvement over StandardGreedy, (3) FastGreedy further
reduces the number of queries of ThresholdGreedy while sacrificing
little in solution quality, and (4) the non-submodularity parameters
on a small instance are computed, which provides evidence that our
theoretical performance ratios are useful.
\paragraph{Quality of Solution}
In Fig. \ref{fig:grqc-qos}, we plot $\mathbb{A}(\vec g)$ for the
solution returned by each algorithm on the GrQc network with $10$ incentive
levels; the difference in quality of solution returned by the three algorithms 
is negligible. In Fig. \ref{fig:fb-qos}, we plot the same for the Facebook
network with $100$ incentive levels; on Facebook, we drop StandardGreedy due to its
prohibitive runtime. FastGreedy is observed to lose a small (up to 3\%)
factor, which we consider acceptable in light of its large runtime improvement,
which we discuss next.
\paragraph{Number of Queries}
Next, we present in Fig. \ref{fig:queries} the
number of function queries\footnote{Our implementation is in terms
  of the marginal gain. The number of function queries shown is
  the number of times the marginal gain function was called.} 
each algorithm requires 
on the GrQc and Facebook networks.
StandardGreedy required up to 20M queries on Facebook,
hence it is not shown in Fig. \ref{fig:q-fb}. 
Both of our algorithms provide a large improvement over
StandardGreedy; in particular, notice that StandardGreedy increases
linearly with $k$, while both of the others exhibit
logarithmic increase in agreement with the theoretical
query complexity of each. Furthermore, FastGreedy uses at least 14.5\% fewer
function queries than ThresholdGreedy and up to 43\% fewer as $k$ grows.

In terms of CPU runtime, we show in Fig.
\ref{fig:runtime} that FastGreedy significantly outperforms
ThresholdGreedy; hence, the runtime appears to be dominated
by the number of function queries as expected. 
\paragraph{Non-Submodularity Parameters}
The value of the FastGreedy DR ratio 
$\beta^*$ on GrQc is shown in Fig. \ref{fig:beta};
notice that it is relatively stable as the budget increases from $K=20$ to $100$, although
there is substantial drop from $10$ incentive levels to $100$; this may be explained as an increase 
in the non-submodularity
resulting from inaccurate sampling of $\mathbb A$, since it is more difficult to detect 
differences between the finer levels. Still, on all instances tested,
$\beta^* > 0.6$, which suggests the worst-case performance ratio of FastGreedy is not far from
that of StandardGreedy.

Finally, we examine the various non-submodularity parameters on a very small instance which admits
their computation: a random Barabasi-Albert network with $10$ nodes and $10$ incentive levels. 
We compute the FastGreedy version of the submodularity ratio $\gamma_s$ defined
in Appendix \ref{apx:nonsm-defs}
by direct enumeration and consider the FastGreedy DR ratio $\beta^*$.
Results are shown in Fig. \ref{fig:gamma}. The value of $\beta^*$ is close to $1$ and remains
constant with increasing budget $k$, while
the FastGreedy submodularity ratio decreases slowly with $k$.
With $\beta^*$ and the FastGreedy $\gamma_s$, we can compute the
worst-case performance ratio of FastGreedy across these instances:
$0.449692$.


\section{Conclusions}\label{sect:conclusion}
In this work, we provide two approximation algorithms for
maximizing non-submodular functions with respect to a cardinality
constraint on the integer lattice with polynomial query complexity. 
Since set functions are a special case, our work provides faster algorithms for the same problem with
set functions than the standard greedy algorithm, although the
performance ratio degrades from at least $1 - e^{-\gamma_s}$ to $1 - e^{-\beta^*\gamma_s}$,
where $\beta^*$ is the FastGreedy DR Ratio.
We propose a natural application of non-submodular influence 
maximization, for which we lower bound the relevant non-submodularity parameters
and validate our algorithms.
\clearpage
\bibliography{mendeley,exp}
\bibliographystyle{plainnat}
\onecolumn
\clearpage
\appendix
\section{Organization of the Appendix}
Appendix \ref{apx:nonsm-defs} defines the greedy versions of the non-submodularity parameters.

Appendix \ref{apx:dr-ratio} provides omitted proofs from Section \ref{sect:dr-ratio}.

Appendix \ref{apx:bspivot} defines the BinarySearchPivot procedure and omitted proofs
from Section \ref{sect:threshold-greedy}.

Appendix \ref{apx:fg} provides omitted proofs from Section \ref{sect:fg}.

Appendix \ref{apx:im} defines the Independent Cascade model, proves that classical IM and boosting are subproblems of our IM model, and provides the proof of Theorem \ref{thm:dr-lower-bound} from Section \ref{sect:IM}.

Appendix \ref{apx:experiments} provides additional experimental results characterizing the parameters of FastGreedy.

Appendix \ref{apx:implementation} presents details of our GIM implementation.
\section{Greedy Versions of Non-Submodularity Parameters}
\label{apx:nonsm-defs}
We define various greedy versions of the non-submodularity parameters
in this section.
In this work, these
are referred to as FastGreedy submodularity ratio, \textit{etc.}, 
where the instance is clear from the context.
\paragraph{ThresholdGreedy DR ratio.}
\begin{deff}[ThresholdGreedy DR ratio]
  Let an instance $\mathcal I$ of Problem \ref{prob:max} be given,
  with budget constraint $k$. 
  Let $\vec g^1, \ldots, \vec g^T$ be the sequence of values $\vec g$ takes during
  execution of ThresholdGreedy on $\mathcal I$. 
  The \textit{ThresholdGreedy version of the DR ratio  on $\mathcal I$}
  $\gamma_d^{\text{TG}, \mathcal I}(f) \in [0,1]$, is the maximum value such that
  for any $i \in \{1, \ldots, T\}$, for any $s \in S$,
  if $\vec g^{i,s}$ is the value of the greedy vector immediately after
  $s$ was considered during the inner \textbf{for} loop of the threshold
  directly preceding the one in which $\vec g^i$ was considered ($\vec g^{i,s} = \vec 0$ if $\vec g^i$
  was considered during the first threshold),
  \[ \gamma_d^{\text{TG}, \mathcal I} \delta_{\vec s} (\vec g^i ) \le
    \delta_{\vec s} (\vec g^{i,s}).\]
\end{deff}
\paragraph{Greedy versions of submodularity ratio.}
\begin{deff}
  Let $\mathcal A \in \{ \text{StandardGreedy}, \text{ThresholdGreedy} \}$, 
  and let an instance $\mathcal I$ of Problem \ref{prob:max} be given,
  with budget constraint $k$. 
  Let $\vec g^1, \ldots, \vec g^T$ be the sequence of values $\vec g$ takes during
  execution of $\mathcal A$ on $\mathcal I$. 
  The \textit{$\mathcal A$ version of the submodularity ratio  on $\mathcal I$}
  $\gamma_s^{\mathcal A, \mathcal I} \in [0,1]$, is the maximum value such that
  for any $s \in S$, for any $i \in \{1, \ldots, T\}$, for any $\vec w$ such that $\vec g^i \le \vec w$ and
  $\lone{\vec w - \vec g^i} \le k$ and $\vec w \leq \vec b$,
  \[ \gamma_s^{\mathcal A, \mathcal I} \left[ f( \vec w) -
    f( \vec g^i) \right] \le \sum_{s \in \{ \vec w - \vec g^i \} }
  \delta_{\vec s} ( \vec g^i ). \]
\end{deff}
The \emph{FastGreedy submodularity ratio} differs from the above two only in that the sequence
of vectors $\vec g^1, \ldots, \vec g^T$ are the value of the greedy vector $\vec g$ at
the beginning of each iteration of the outer \textbf{while} loop, instead of all values
of $\vec g$ during execution of the algorithm.
\section{Proofs for Section \ref{sect:dr-ratio}} \label{apx:dr-ratio}
\begin{proof}[Proof of Proposition \ref{prop:sm-ratio}]
  Suppose $\vec v \le \vec w \in \lattice$.  
  Let $\{\vec w - \vec v \} = \{s_1,...,s_l\}$. Then,
  \begin{align*}
    \gamma_d(f(\vec w) - f(\vec v)) & =
    \gamma_d\sum_{i=1}^l[f(\vec v + \vec s_1 + ... + \vec s_i)
      - f(\vec v + \vec s_1 + ... + \vec s_{i-1})] \\
      & = \gamma_d \sum_{i = 1}^l \delta_{\vec s_i}(\vec v + \vec s_1 + ... + \vec s_{i-1}) \\
      & \leq \sum_{i=1}^l \delta_{\vec s_i}(\vec v)
  \end{align*}
  Therefore, $\gamma_d \leq \gamma_s$, since $\gamma_s$ is the maximum number satisfying
  the above inequality.
\end{proof}

\section{BinarySearchPivot and Proofs for Section \ref{sect:threshold-greedy} (ThresholdGreedy)} \label{apx:bspivot}
\paragraph{BinarySearchPivot.}
\begin{algorithm}[tb]
  \caption{BinarySearchPivot$(f, \vec g, \vec b, s, k, \tau)$}\label{alg:bspivot}
  \begin{algorithmic}[1]
    \STATE {\bfseries Input:} $f \in \funcs_{\vec b}$, $\vec g \in \lattice$, $\vec b \in \lattice$, 
    $s \in S$, $k \in \nats$, $\tau \in \reals$
    \STATE {\bfseries Output:} $l \in \mathbb{N}$
    \STATE $l_s \gets 1, l_t \gets \min \{ \vec b_s - \vec g_s, k - \lone{\vec g} \}$,
    \IF{ $\delta_{l_t \vec s}( \vec g ) \ge l_t\tau$\label{line:fp-max}}
    \STATE \textbf{return} $l_t$\label{fp:max}
    \ENDIF
    \IF{$\delta_{\vec s} ( \vec g ) < \tau$}
    \STATE \textbf{return} $0$\label{fp:0}
    \ENDIF
    \WHILE { $l_t \neq l_s + 1$ }
    \STATE $m = \lfloor (l_t + l_s) / 2 \rfloor$
    \IF{$\delta_{m \vec s} ( \vec g ) \ge m \tau$}
    \STATE $l_s = m$
    \ELSE
    \STATE $l_t = m$
    \ENDIF
    \ENDWHILE
    \STATE \textbf{return } $l_s$
\end{algorithmic}
\end{algorithm}
The routine BinarySearchPivot (Alg. \ref{alg:bspivot}) efficiently finds a pivot for each
$s \in S$.
BinarySearchPivot
 uses a modified binary-search procedure that maintains $l_s <l_t$
such that both
\begin{align}
  \delta_{l_s \vec s}( \vec g) &\ge l_s \tau, \label{ineq:1}\\
  \delta_{l_t \vec s}(\vec g) &< l_t \tau. \label{ineq:2}
\end{align}
Initially, $l_s$ and $l_t$ do satisfy (\ref{ineq:1}), (\ref{ineq:2}), 
or else we have already found a valid pivot (lines \ref{fp:max}, \ref{fp:0}). 
The midpoint $m$ of the interval $[l_s, l_t]$ is tested to determine if
$l_s$ or $l_t$ should be updated to maintain (\ref{ineq:1}), (\ref{ineq:2});
this process continues until $l_t = l_s + 1$.
\begin{lemma} \label{lemm:pivot}
  BinarySearchPivot finds a valid pivot $l \in \{0, \ldots, l_{max} \}$
  in $O( \log l_{max} )$ queries of $f$, where
  $l_{max} = \min \{ \vec b_s - \vec g_s, k - \lone{\vec g} \}$,
  $\vec b_{max} = \max_{s \in S} \vec b_s$.
\end{lemma}
\begin{proof}[Proof of Lemma \ref{lemm:pivot}]
  The routine BinarySearchPivot 
  maintains inequalities (\ref{ineq:1}), (\ref{ineq:2}), 
  it is enough to show that given (\ref{ineq:1}), (\ref{ineq:2}), 
  there exists a
  $l \in \{l_s, ..., l_t - 1 \}$ such that $l$ is a pivot.
  Consider $l_j = l_t - j$, for $j \in \nats$; there must be a smallest
  $j \ge 1$ such that $l_j$ satisfies property (\ref{prop:1}) of being
  a pivot, since $l_s < l_t$ satisfies property (\ref{prop:1}). 
  If property (\ref{prop:2}) is unsatisfied, then
  \begin{align*}
    \delta_{(l_j + 1)e_i} ( s ) &= \delta_i( s + l_j ) + \delta_{l_je_i} ( s ) \\
    &\ge \tau + l_j\tau = (l_j + 1)\tau,
  \end{align*}
  contradicting the choice of $j$ since $l_j + 1 = l_{j - 1}$.
  The query complexity follows from a constant number of queries per
  iteration of the while loop and the fact that each iteration reduces
  the distance from $l_s$ to $l_t$ by a factor of 2; initially, this
  distance was $l_{max}$.
\end{proof}
\paragraph{Omitted proofs from Section \ref{sect:threshold-greedy}.}
\begin{proof}[Proof that Property \ref{lemm:exist} holds]
  Let $\vec g^{\tau,s}$ be the value of $\vec g$ immediately after 
  $s$ is considered during the iteration corresponding to $\tau$; 
  then property (\ref{prop:2}) of pivot was satisfied: $\delta_{\vec s}( \vec g^{\tau,s} ) < \tau$.
\end{proof}
\begin{proof}[Proof of Claim \ref{lemm:wlog}]
  Suppose
  $\gamma_d \ge \epsi$.
  Suppose
  $\lone{\vec g} < k$, and let $\vec g'$ be the solution returned by a modified
  ThresholdGreedy that continues updating the threshold until
  $\lone{\vec g'} = k$. Order $\{ \vec g' \} \setminus \{ \vec g \} = \{ \vec s_1, \ldots, \vec s_\ell \}$,
  and let $\vec g'_i = \vec g'_{i - 1} + \vec s_i$, $i = 1, \ldots, \ell$, with
  $\vec g'_0 = \vec g$, so that $\vec g'_\ell = \vec g'$. Also, let $\vec g_{i-1} \le \vec g$ 
  be the vector guaranteed for $\vec s_i$ 
  by Lemma \ref{lemm:exist} with the last threshold value
  $\tau$ of ThresholdGreedy.
  Then
  \begin{align*} 
    f(\vec g') - f(\vec g) &= \sum_{i=1}^\ell \delta_{\vec s_i} ( \vec g'_{i - 1} ) \\
                           &\le \frac{1}{\gamma_d} \sum_{i = 1}^\ell \delta_{\vec s_i} (\vec g_{i - 1} ) \\
                           &\le \frac{\ell}{\epsi} \frac{\epsi ^2M}{k} \le \epsi M \le \epsi f( \vec \Omega ).
  \end{align*}
  Hence, for any $\Phi > \epsi$, if 
  $$f( \vec g' ) \ge \Phi f( \vec \Omega ),$$
  then 
  $$f( \vec g ) \ge ( \Phi - \epsi ) f( \vec \Omega ).$$
\end{proof}

\paragraph{From proof of Theorem \ref{thm:threshold}:}

\paragraph{``If $\gamma_d < \epsi$, the ratio holds trivially''.}
If $\gamma_d < \epsi$, the ratio holds trivially
  from the 
  inequality $1 - e^{-x} \le x$, for real $x > 0$, 
  since 
  $$1 - e^{-\gamma_d\gamma_s\kappa} \le \gamma_d \gamma_s \kappa < \epsi.$$

\paragraph{``from which the hypothesis of Claim \ref{lemm:wlog} follows''.}
Since $(1 - x) \le e^{-x}$ and $\sum_{t} l^t = k$, we have
  $\prod_{t = 1}^T (1 - l^t \gamma_d \gamma_s \kappa / k ) \le \prod_{t = 1}^T \exp ( (- l^t \gamma_d\gamma_s \kappa / k ) ) = \exp ( - \gamma_d \gamma_s \kappa )$.

\begin{proof}[Proof of Corollary \ref{cor:thresh-par}]
  As in proof of Theorem \ref{thm:threshold}, suppose $\gamma_d \ge \epsi$.
  Claim \ref{lemm:wlog} still holds as before. Now, let $\vec g^t$ be the value
  of $\vec g$ at the beginning of the $t$th iteration of the outer \textbf{for} loop with
  threshold value $\tau_t$.
  Since the inner \textbf{for} loop is conducted in parallel, all marginal gains in
  iteration $t$ are considered with respect to $\vec g^t$. Order the vectors added 
  in this iteration $l_1\vec s_1, \ldots, l_\ell \vec s_\ell$; because each $l_i$ is a pivot,
  we know $\delta_{l_i\vec s_i}( \vec g^t) \ge l_i \tau_t$ and $\delta_{\vec s_i}(\vec g^t + l_i \vec s_i) < \tau_t$.

  Let $\vec g_i^t = \vec g_{i -1}^t + l_i \vec s_i$, so $\vec g_0^t = \vec g^t$ and $\vec g_\ell^t = \vec g^{t + 1}$. Now for each $i$ and for each $s \in S$, there exists a vector $\vec h_i^s \le \vec g^t$ such that $\delta_{\vec s}( \vec h_i^s ) < \tau_t / \kappa$ (namely $\vec h_i^s = \vec g^{t-1} + l^*\vec s_i$, from when $\vec s_i$ was considered during the previous iteration $t - 1$, or $\vec h_i^s = \vec 0$ if $t = 1$ is the first iteration). Furthermore $\vec g^t \le \vec g_i^t$ and $\delta_{\vec l_i s_i}( \vec g^t ) \ge l_i \tau_t$. Hence $$ \delta_{l_i \vec s_i}( \vec g_i^t ) \ge (1 - \alpha) \delta_{l_i \vec s_i}( \vec g^t ) \ge (1 - \alpha)l_i \tau_t \ge \kappa (1 - \alpha) l_i \delta_{\vec s}( \vec h_i^s) \ge \kappa (1 - \alpha) l_i \gamma_d \delta_{\vec s}( \vec g_i^t ), $$
for any $s \in S$. The preceding argument proves an analogue of Claim \ref{lemm:marge}, and
the argument from here is exactly analogous to the proof of Theorem \ref{thm:threshold}.
\end{proof}
\section{Proofs for Section \ref{sect:fg}} \label{apx:fg}
\begin{proof}[Proof of Theorem \ref{thm:fg-ratio}]
Since we have included $\gamma_d \ge \epsi$ as a hypothesis,
we have the following claim, analogous to Claim \ref{lemm:wlog}.
  \begin{claim} \label{claim:fg-wlog}
    If $\vec g$ is produced by the modified version of FastGreedy that
    continues until $\lone{ \vec g } = k$, and
    $f( \vec g ) \ge (1 - e^{-\kappa \beta^* \gamma_s})f( \vec \Omega )$,
    then the Theorem is proved.
  \end{claim}  
\begin{proof}
 Suppose
  $\lone{\vec g} < k$, and let $\vec g'$ be the solution returned by a FastGreedy*
  which continues updating the threshold until
  $\lone{\vec g'} = k$. Order $\{ \vec g' \} \setminus \{ \vec g \} = \{ \vec s_1, \ldots, \vec s_\ell \}$,
  and let $\vec g'_i = \vec g'_{i - 1} + \vec s_i$, $i = 1, \ldots, \ell$, with
  $\vec g'_0 = \vec g$, so that $\vec g'_\ell = \vec g'$. 
  Then
  \begin{align*} 
    f(\vec g') - f(\vec g) &= \sum_{i=1}^\ell \delta_{\vec s_i} ( \vec g'_{i - 1} ) \\
                           &\le \frac{1}{\gamma_d} \sum_{i = 1}^\ell \delta_{\vec s_i} (\vec g ) \\
                           &\le \frac{\ell}{\epsi} \frac{\epsi ^2M}{k} \le \epsi M \le \epsi f( \vec \Omega ).
  \end{align*}
  Hence, for any $\Phi > \epsi$, if 
  $$f( \vec g' ) \ge \Phi f( \vec \Omega ),$$
  then 
  $$f( \vec g ) \ge ( \Phi - \epsi ) f( \vec \Omega ).$$ 
\end{proof}
Thus, for the rest of the proof, let $\vec g$ be produced by the modified version of
FastGreedy as in the hypothesis of Claim \ref{claim:fg-wlog}.
  Let $\vec s^t \in S$, $\vec g^t$ be the direction maximizing the marginal gain on 
  line \ref{line:greedy}, the solution $\vec g$ 
  immediately after the $t$th iteration of the \textbf{while} loop, respectively.
  By the choice of $\vec s^t$, for each $s \in \{ \vec \Omega \}$, 
  we have $\delta_{\vec s} ( \vec g^{t - 1} ) \le \delta_{\vec s^t} ( \vec g^{ t - 1 } )$.
  Let $l_1\vec s_1',\ldots,l_\ell \vec s_\ell'$ be the additions 
  on line \ref{line:greedyAdd} to the solution $\vec g$ during iteration $t$, 
  with each $l_m > 0$ for $m = 1, \ldots, \ell$.
  Let $\vec g^{t-1}_0 = \vec g^{t-1}$ and $\vec g^{t-1}_m = \vec g^{t-1}_{m-1} + l_m\vec s_m'$. 
  Let $L_t = \sum_{m=1}^\ell l_m$.
  Now, $l_m$ was chosen by BinarySearchPivot and hence satisfies
  $\delta_{l_m\vec s_m'} ( \vec g^{t - 1}_{m-1} ) \ge l_m \beta \kappa \delta_{\vec s^t} ( \vec g^{t - 1} )$ by property (\ref{prop:1}) of pivot and the choice of threshold $\tau = \beta \kappa m$.
  So 
\begin{align*}
    f( \vec g^t ) - f( \vec g^{t-1} ) &= \sum_{m=1}^\ell f( \vec g^{t-1}_m ) - f( \vec g^{t-1}_{m-1}) \\ 
                            &\ge \sum_{m=1}^\ell l_m \beta \kappa \delta_{\vec s^t} ( \vec g^{t-1} )\\
                            &= L_t \beta \kappa \delta_{\vec s^t} ( \vec g^{t-1} )\\
                            &\ge \frac{L_t \beta^*\kappa}{k} \sum_{s \in \{ \vec \Omega - (\vec g^{t-1} \wedge \vec \Omega ) \}} \delta_s ( \vec g^{t - 1} ) \\
    &\ge \frac{ L_t \beta^* \gamma_s \kappa }{k} \left( f( \vec \Omega ) - f( \vec g^{t-1} ) \right),\\
  \end{align*}
  where the first inequality is by definition of $\vec g_m^t$, the first
  inequality is by the preceding paragraph, the second equality is by definition
  of $L_t$, the second inequality is by the selection of $\vec s^t$ and that 
  fact $\lone{ \vec \Omega } \le k$, and the third inequality is by the definition
  of submodularity ratio and the lattice identity $\vec v \vee \vec w - \vec v = \vec w - \vec v \wedge \vec w$.
  From here,
  \[ f( \vec g ) \ge \left( 1 - \prod_{t=1}^T \left( 1 - \frac{L_t \beta^* \gamma_s \kappa}{k} \right) \right) f( \vec \Omega ), \]
  from which the hypothesis of Claim \ref{claim:fg-wlog} follows:
  since $(1 - x) \le e^{-x}$ and $\sum_{t} L^t = k$, we have
  $\prod_{t = 1}^T (1 - L^t \beta^* \gamma_s \kappa / k ) \le \prod_{t = 1}^T \exp ( (- L^t \beta^*\gamma_s \kappa / k ) ) = \exp ( - \beta^* \gamma_s \kappa )$.
\end{proof}
  \begin{proof}[Proof of Claim \ref{claim:fast}] 
  For any $i$, $m'_i \le m'_{i - 1} / \gamma_d$: to see this, observe
  $m'_i = \max_{s \in S} \delta_{\vec s} ( \vec g^i )$,
  $m'_{i - 1} = \max_{s \in S} \delta_{\vec s} ( \vec g^{i - 1} )$,
  for some $\vec g^{i - 1} \le \vec g^i$. For each $s \in S$,
  $\delta_{\vec s} ( \vec g^{i - 1} ) \le m'_{i - 1}$, so 
  $\delta_{\vec s} ( \vec g^{i } ) \le m'_{i - 1} / \gamma_d$,
  and hence so is $m_i'$.
  Since
  $j_\ell$ is the last uptick in the sequence before the deletion, we know for 
  every $i > j_\ell$, $m'_i \le \kappa m'_{i-1}$. Hence the deleted sequence
  proceeds from $m'_{j_\ell + 1} \le m'_\ell / \gamma_d$ down to 
  $m'_{k_\ell - 1} \ge \kappa m'_{j_\ell}$ by definition of $m'_{k_\ell}$, with 
  each term decreasing by
  a factor of at least $\kappa$. 
  \end{proof}
\section{Influence Maximization: A General Framework} \label{apx:im}
\paragraph{``explicit formula for $p^{\vec x}(H,T)$''.}
$p^\vec x (H,T) = Pr ( H | \vec x ) Pr ( T | \vec x )$, with
\begin{align*}
  Pr \left(H | \vec x \right) &= \prod_{(x,y) \in E} p(x,y, \vec x_y)^{\ind{(x,y) \in H}}( 1 - p(x,y, \vec x_y))^{\ind{(x,y) \not \in H}}, \\
  Pr \left(T | \vec x \right) &= \prod_{y \in V} p(y,\vec x_y)^{\ind{y \in T}}( 1 - p(y, \vec x_y))^{\ind{y \not \in T}}.
\end{align*}

\paragraph{The Independent Cascade (IC) Model.}
The IC model is defined as follows. Given a graph $G = (V,E)$, with probabilities $p(e)$
associated to each edge $e \in E$. Let $H$ be a realization of $G$, where each edge $e$
is included in $H$ with probability $p(e)$. Then, from an initial seed set $T$ of
activated users, a user is activated if it is reachable in $H$ from $T$. Intuitively,
the weight on edge $(u,v)$ represents the probability that $u$ activates $v$ (\textit{i.e} user $u$ convinces $v$ to adopt the product).  For more information, we refer the reader to \citet{Kempe2003}.
\begin{proposition} \label{prop:im}
  There is a natural one-to-one correspondence between
  instances of the IM problem under IC model
  and a subclass of instances of GIM.
\end{proposition}
\begin{proof}[Proof of Proposition \ref{prop:im}]
With exactly two levels, our GIM can encapsulate the 
classical IM problem with the IC model \citep{Kempe2003}.
Let weighted social network $G= (V,E)$ and budget $k$ be given,
as an instance of the IM problem. This instance corresponds
to one of GIM with the same network and budget, as follows. 
For each edge $(u,v) \in E$,
let $w(u,v)$ be its weight. Then we assign $p(u,v,i) = w(u,v)$ for
each $i \in \{0,1\}$. 
Each incentive vector $\vec x$ is a binary vector, indicating
which users are present in the seed set; \ie $p(u,0)=0$ and
$p(u,1) = 1$, for all $u \in V$.
This mapping is injective and hence invertible.
\end{proof}
\begin{proposition} \label{prop:boost}
  There is a natural one-to-one correspondence between
  instances of the boosting problem 
  and a subclass of instances of GIM.
\end{proposition}
\begin{proof}[Proof of Proposition \ref{prop:boost}]
Let social network $G = (V,E)$, seed set $S$, and $k \in \nats$
be given as an instance of the boosting problem, where edge
$(u,v) \in E$ has weight $p(u,v)$ if $v$ is not boosted,
and weight $p'(u,v)$ if $v$ is boosted.
The corresponding instance of GIM has two levels.
Set $p(y,0) = p(y,1) = 1$ for all $y \in S$ and set $p(y,0)=p(y,1)=0$ for all
$y \not \in S$. For each edge $(u,v) \in E$, set $p(u,v,0) = p(u,v)$, $p(u,v,1)=p'(u,v)$.
Hence, spending budget to incentivize a node from level 0 to level 1
does not affect the initial seed set, which is always $S$. But this incentive
does work in exactly the same way as the boosting of a node by changing its
incoming edge probabilities; hence, the objective values are the same. 
This mapping is injective and hence invertible.
\end{proof}
\begin{proof}[Proof of Theorem \ref{thm:dr-lower-bound}]
\begin{claim} \label{clm:suffices3}
  Suppose $\gamma \delta_{\vec s}( \vec w ) \le \delta_{\vec s}( \vec v )$,
  where $\gamma = c_e^{-k\Delta}c_n^{-k}$, and $\vec v,\vec w$ are any vectors
  satisfying $\vec v \le \vec w$, $\lone{ \vec v }\le k$, and $\lone{\vec w - \vec v} \le k$.
  Then the result of Theorem \ref{thm:dr-lower-bound} follows.
\end{claim}
\begin{proof}
  Suppose the hypothesis of the claim holds. 
  \paragraph{Greedy submodularity ratios.}
  We will show $\gamma \le \gamma_s^{\mathcal A, \mathcal I}$, where $\gamma_s^{\mathcal A, \mathcal I}$
  is the FastGreedy submodularity ratio on instance $\mathcal I$. The proofs for the other
  greedy submodularity ratios are exactly analogous. 

  Let $\vec g^1, \ldots, \vec g^T$ be the greedy vectors in the definition of
  FastGreedy submodularity ratio. Let $i \in \{1, \ldots, T\}$,
  and let $\vec v = \vec g^i$. 
  Let $\vec w \ge \vec v \in \lattice$ such that
  $\lone{\vec w - \vec v} \le k$. Let $\{\vec w - \vec v \} = \{s_1,...,s_l\}$.
  Then $\vec v + \sum_{j=1}^l \vec s_j = \vec w$. In addition, for
  every $m \leq l$, $\vec v + \sum_{j=1}^m \vec s_j = \vec v_m$
  where $\vec v_m \in \lattice$, $\vec v_m \leq \vec w$. Then,
  \begin{align*}
    \gamma(f(\vec w) - f(\vec v)) & =
      \gamma \sum_{j=1}^l f(\vec v + \vec s_1 + ... + \vec s_j) -
      f(\vec v + \vec s_1 + ... + \vec s_{j-1}) \\
    & = \gamma \sum_{j=1}^l \delta_{\vec s_j}(\vec v + \vec s_1 + ... + \vec s_{j-1}) \\
    & = \gamma \sum_{j=1}^l \delta_{\vec s_j}(\vec v_{j-1}) \\
    & \leq \sum_{j=1}^l \delta_{\vec s_j}(\vec v),
  \end{align*}
  by the hypothesis of the claim.
  Therefore $\gamma \leq \gamma_s^{\mathcal A, \mathcal I}$, 
  since the latter is the maximum number satisfying
  the above inequality for each $\vec w, \vec g^i$ as above.
  \paragraph{FastGreedy DR ratio $\beta^*$.}
  Initally, $\beta = 1$; it decreases by a factor of $\delta \in (0,1)$
  at most once per iteration of the \textbf{while} loop of FastGreedy.
  Suppose $\beta \le \gamma$ for some iteration 
  $i$ of the \textbf{while} loop, and let $\vec g$ have the value 
  assigned immediately after iteration $i$, $m$ have the value
  assigned after line \ref{line:greedy} of iteration $i$. Then 
  Since a valid pivot was found 
  for each $s \in S$ during iteration $i$, by property (\ref{prop:2}) there exists
  $\vec g^s \le \vec g$, $\delta_{\vec s} ( \vec g^s ) < \beta \kappa m \le \gamma \kappa m$.
  Hence $\delta_{\vec s} (\vec g) \le \kappa m $, 
  since $\vec g$, $\vec g^s$ are vectors satisfying the conditions on $\gamma$
  in the hypothesis of the claim.
  In iteration $i + 1$, $m'$ has the value of $m$ from iteration $i$,
  so the value of $m$ computed during iteration $i + 1$ is at 
  most $\kappa m'$, and $\beta$ does not decrease during iteration $i$.
  It follows that $\beta^* \ge \gamma \delta$.
\end{proof}
Let  $\vec v \le \vec w$, $\lone{ \vec{ v - w } } \le k$.
We will consider graph realizations $H$ that have the status of all edges
determined; and seed sets $T \subseteq V$. 

Then,
\begin{align}
  p^{\vec w + \vec{s} }( H, T ) &= K_1(H, T) p^{\vec v + \vec{s}}( H, T ), \text{ and}\\
  p^{\vec w}( H, T ) &= K_2(H,T) p^{\vec v}( H, T ),
\end{align}
where 
$K_{1}(H, T) = K_1(H) K_1(T)$,
$K_2(H,T) = K_2(H)K_2(T)$ with 
\begin{align*} K_2 ( H ) &= \prod_{(x,y) \in E} \left( \frac{ p(x,y,\vec{w}_y)}{ p(x,y,\vec{v}_y )} \right)^{\ind{ (x,y) \in H}} \left( \frac{ 1 - p(x,y, \vec{w}_y )}{ 1 - p(x,y, \vec{v}_y )} \right)^{\ind{ (x,y) \not \in H}},\\
K_2(T) &= \prod_{x \in V} \left( \frac{ p(x, \vec{w}_x )}{ p( x, \vec{v}_x )} \right)^{\ind{x \in T}} \left( \frac{1 - p(x, \vec{w}_x)}{1 - p( x, \vec{v}_x) } \right)^{\ind{x \not \in T}},
\end{align*} 
and the definitions of $K_1(H),K_1(T)$ are analogous to the above with 
vectors $\vec{v + s}, \vec{w + s}$ in place of $\vec{v},\vec{w}$.
\begin{lemma}\label{lemm:gamma-bound} Let $\Delta$ be the maximum in-degree
  in $G$. 
\[ K_1 (H,T) \le K_2(H,T) \le c_e^{k\Delta} c_n^{k}\]
\end{lemma}
\begin{proof}
  \emph{(a)} $K_1( T ) \le K_2( T )$: if $s \in T$,
  \begin{align*}
    K_1( T ) &= K_2( T ) \cdot \frac{p(s, \vec{v}_s)}{p(s, \vec{w}_s)} \cdot \frac{p(s, \vec{w}_s + 1)}{p(s, \vec{v}_s + 1)}\\ 
             &=  K_2( T ) \cdot \frac{z}{z'} \cdot \frac{z' + \alpha'}{z + \alpha} \le K_2(T),
  \end{align*}
  where $\alpha' \le \alpha$ by DR-submodularity of $i \mapsto p(s, i)$, and 
  $p(s, \vec{v}_s) = z \le z'=p(s,\vec{w}_s)$ by monotonicity of the
  same mapping. 
  Otherwise, if $s \not \in T$, 
  \begin{align*}
    K_1( T ) &= K_2( T ) \cdot \frac{1 - p(s, \vec{v}_s)}{1 - p(s, \vec{w}_s)} \cdot \frac{1 - p(s, \vec{w}_s + 1)}{1 - p(s, \vec{v}_s + 1)}\\ 
             &=  K_2( T ) \cdot \frac{z}{z'} \cdot \frac{z' - \alpha'}{z - \alpha} \le K_2(T),
  \end{align*}
  where as before $\alpha' \le \alpha$ by DR-submodularity, 
  but $1 - p(s, \vec{v}_s) = z \ge z'= 1 - p(s,\vec{w}_s)$.

  \emph{(b)} $K_1(H) \le K_2(H)$: 
  if $(u,s) \in H$,
  \begin{align*}
    K_1( H ) &= K_2( H ) \cdot \frac{p(u,s, \vec{v}_s)}{p(u,s, \vec{w}_s)} \cdot \frac{p(u,s, \vec{w}_s + 1)}{p(u, s, \vec{v}_s + 1)}\\ 
             &=  K_2( H ) \cdot \frac{z}{z'} \cdot \frac{z' + \alpha'}{z + \alpha} \le K_2(H),
  \end{align*}
  where $\alpha' \le \alpha$ by DR-submodularity of $i \mapsto p(u,s, i)$, and 
  $p(u,s, \vec{v}_s) = z \le z'=p(u,s,\vec{w}_s)$ by monotonicity of the
  same mapping. 
  Otherwise, if $(u,s) \not \in H$, 
  \begin{align*}
    K_1( H ) &= K_2( H ) \cdot \frac{1 - p(u,s, \vec{v}_s)}{1 - p(u,s, \vec{w}_s)} \cdot \frac{1 - p(u,s, \vec{w}_s + 1)}{1 - p(u,s, \vec{v}_s + 1)}\\ 
             &=  K_2( H ) \cdot \frac{z}{z'} \cdot \frac{z' - \alpha'}{z - \alpha} \le K_2(H),
  \end{align*}
  where as before $\alpha' \le \alpha$ by DR-submodularity, 
  but $1 - p(u,s, \vec{v}_s) = z \ge z'= 1 - p(u,s,\vec{w}_s)$.

  \emph{(c)} $K_2(H,T) \le c_e^{k\Delta}c_n^k$: 
  \begin{align*}
    K_2(H,T) = K_2(H)K_2(T) = \prod_{(x,y) \in E} \left( \frac{ p(x,y,\vec{w}_y)}{ p(x,y,\vec{v}_y )}\right)^{\ind{ (x,y) \in H}} &\left( \frac{ 1 - p(x,y, \vec{w}_y )}{ 1 - p(x,y, \vec{v}_y )} \right)^{\ind{ (x,y) \not \in H}} \\ 
                                                                                                      &\prod_{x \in V} \left( \frac{ p(x, \vec{w}_x )}{ p( x, \vec{v}_x )} \right)^{\ind{x \in T}} \left( \frac{1 - p(x, \vec{w}_x)}{1 - p( x, \vec{v}_x) } \right)^{\ind{x \not \in T}}
  \end{align*}
  Each of the fractions in the above product is of the form
  $\xi ( \vec w_y ) / \xi ( \vec v_y )$, where $\vec w_y \ge \vec v_y$ and hence can
  be written
  \begin{align*}
    \frac{\xi ( \vec w_y )}{\xi ( \vec v_y )} = \prod_{i = 1}^{\vec w_y - \vec v_y} \frac{\xi( \vec v_y + i + 1 )}{ \xi (\vec v_y + i )} \le (\vec w_y - \vec v_y) \max_j \frac{\xi (j + 1)}{\xi (j)}.
  \end{align*}
  Hence, by the fact that $\lone{\vec w - \vec v} \le k$ and the definitions of $c_e,c_n$, 
  and the maximum in-degree $\Delta$ in $G$, we have
  $$K_2(H,T) \le c_e^{k \Delta}c_n^k.$$
\end{proof}
Finally, by Lemma \ref{lemm:gamma-bound}, we have
\begin{align*}
  \mathbb{A}( \vec w + \vec s ) - \mathbb A( \vec w ) &= \sum_{H, T} \left( p^{\vec{w + s}} ( H, T ) - p^{\vec{w}}( H, T ) \right)R( H, T ) \\
                                                      &= \sum_{H, T} \left( K_1(H,T) p^{\vec{v + s}} ( H, T ) - K_2(H,T) p^{\vec{v}}( H, T ) \right)R( H, T ) \\
  &\le \sum_{H, T} K_2(H,T) \left( p^{\vec{v + s}} ( H, T ) - p^{\vec{v}}( H, T ) \right)R( H, T ) \\
  &\le {c_e}^{k\Delta}c_n^{k} \left( \mathbb A( \vec{v} + \vec{s} ) - \mathbb A( \vec{v} ) \right).
\end{align*}
Therefore, the hypothesis of Claim \ref{clm:suffices3} is satisfied, and the result follows.
\end{proof}
\section{Additional Experimental Results} \label{apx:experiments}
\subsection{Characterizing the Parameters of FastGreedy}\label{sect:experiments-parameters}
\begin{figure}
    \centering
    \subfigure[$\kappa$]{\includegraphics[width=0.22\textwidth,height=0.13\textheight]{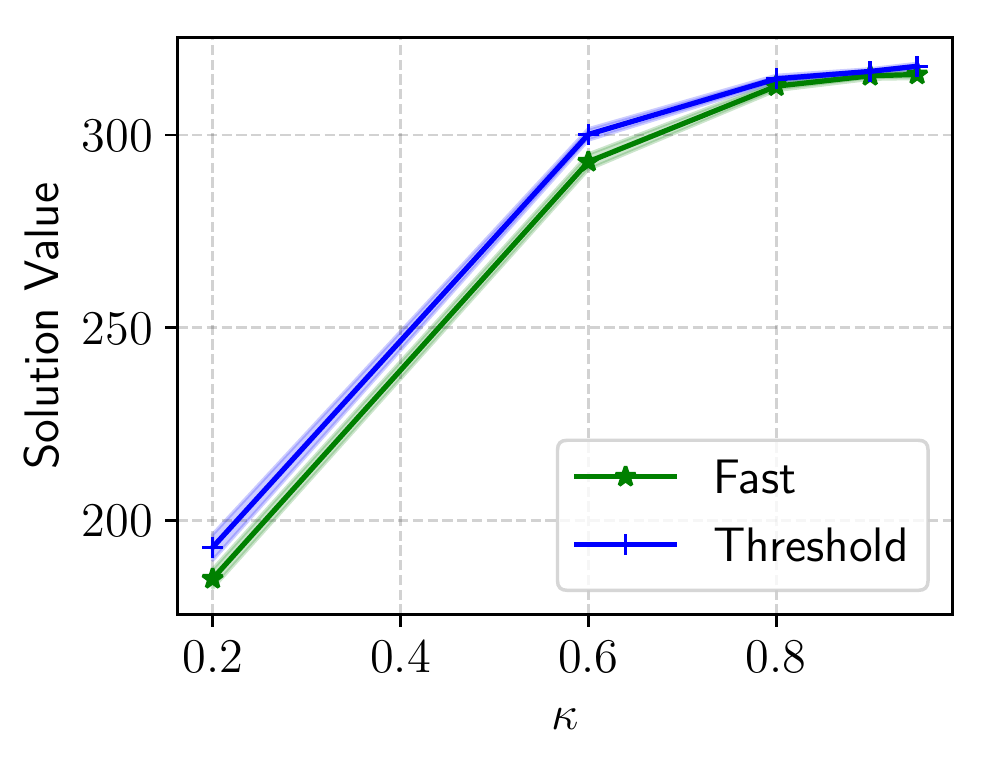}}
    \subfigure[$\delta$]{\includegraphics[width=0.22\textwidth,height=0.13\textheight]{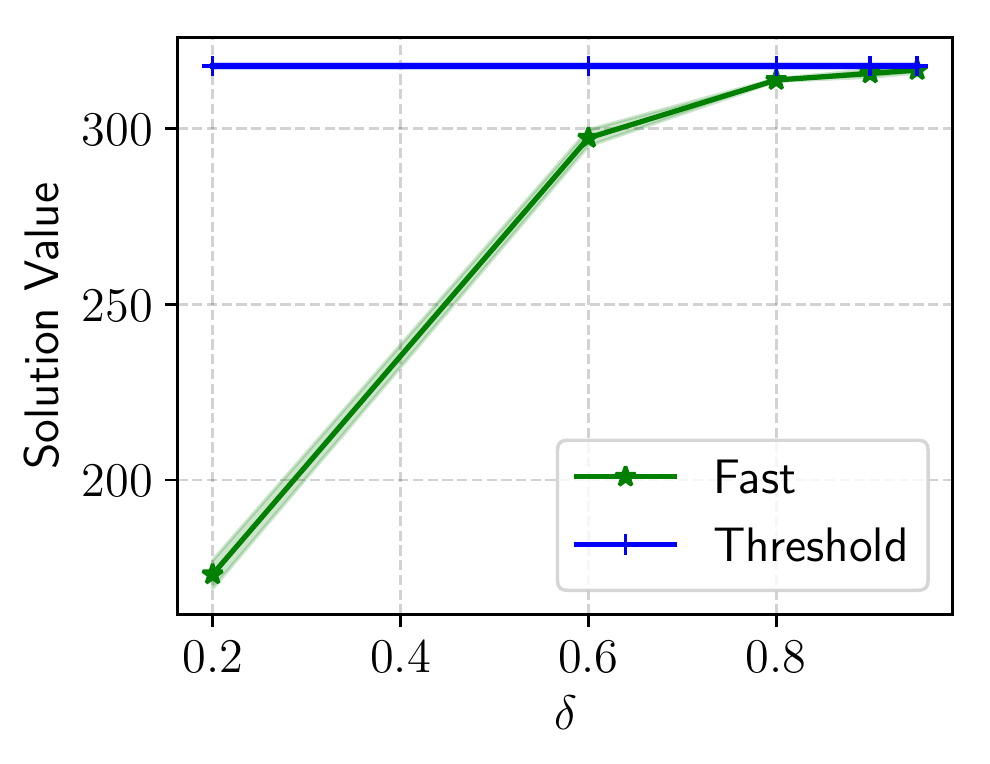}}
    \caption{\label{fig:perf-varying}Performance as $\delta$ and $\kappa$ are varied. Note that ThresholdGreedy does not use $\delta$.}
\end{figure}

\begin{figure}
    \centering
    \subfigure[$\kappa$]{\includegraphics[width=0.22\textwidth,height=0.13\textheight]{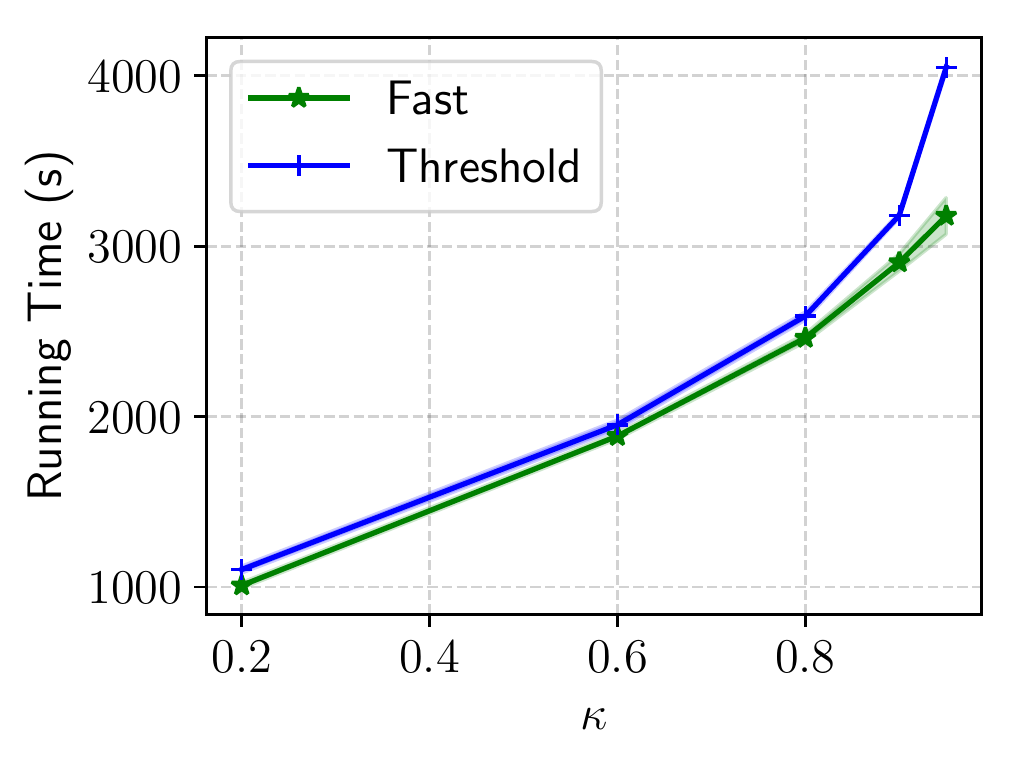}}
    \subfigure[$\delta$]{\includegraphics[width=0.22\textwidth,height=0.13\textheight]{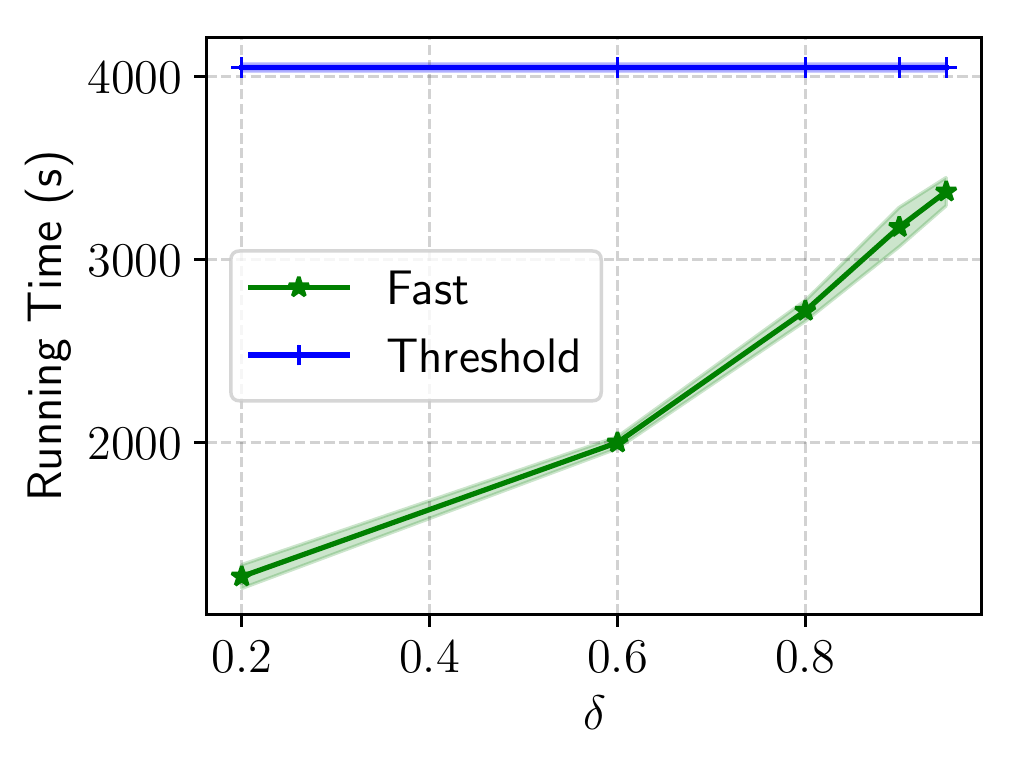}}
    \caption{\label{fig:runtime-varying}Running time as $\delta$ and $\kappa$ are varied. Note that ThresholdGreedy does not use $\delta$.}
\end{figure}
In this section, we evaluate the impact of varying the parameters of FastGreedy: 
$\epsi$, $\delta$, and $\kappa$. We note that $\epsi$ only
impacts the running time and performance if it is used as a stopping
condition. However, in our experiments this did not occur: the
cardinality constraint was reached first. Therefore, in our experiments
FastGreedy = FastGreedy* and we are free to set $\epsi = 0$ without changing
any of our results.

Figures \ref{fig:perf-varying} \& \ref{fig:runtime-varying} show the
impact of $\delta$ \& $\kappa$ on performance and running time. Note
that performance remains similar until $\kappa$ or $\delta$ drops below
\num{0.6}. However, the running time plummets to nearly half of what it
is at \num{0.95} in each case, resulting in a similar-quality solution
in significantly less time. A natural follow-up question from this
figures is: what happens when both parameters are varied at once? Fig.
\ref{fig:kappa-delta} details the answer. In particular, we note that
the steep drop in running time remains present, and there is a
reasonable gain to be had by dropping both parameters at once -- up to a
point.

\begin{figure}
    \centering
    \subfigure[Solution Quality]{\includegraphics[width=0.22\textwidth,height=0.13\textheight]{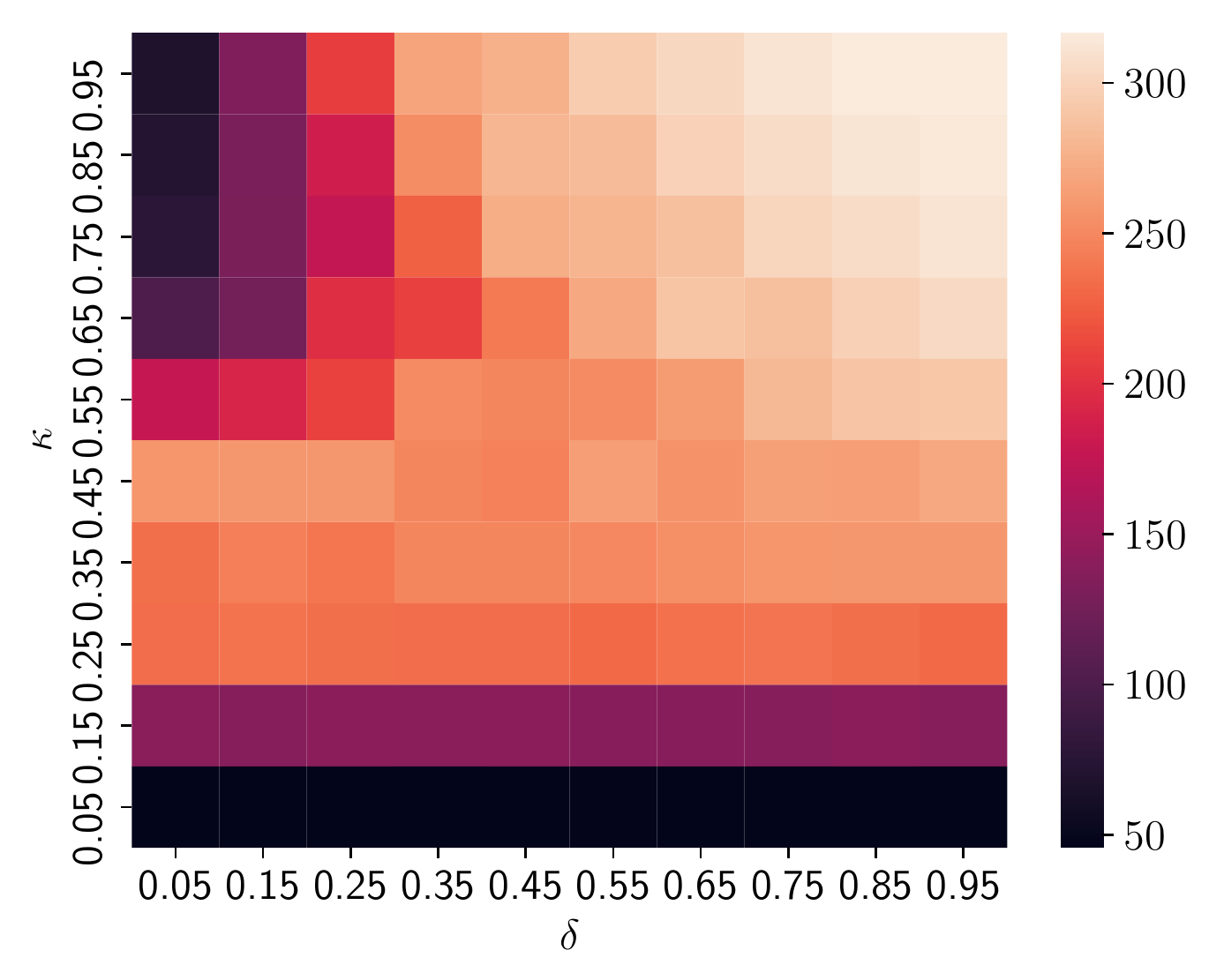}}
    \subfigure[Running Time]{\includegraphics[width=0.22\textwidth,height=0.13\textheight]{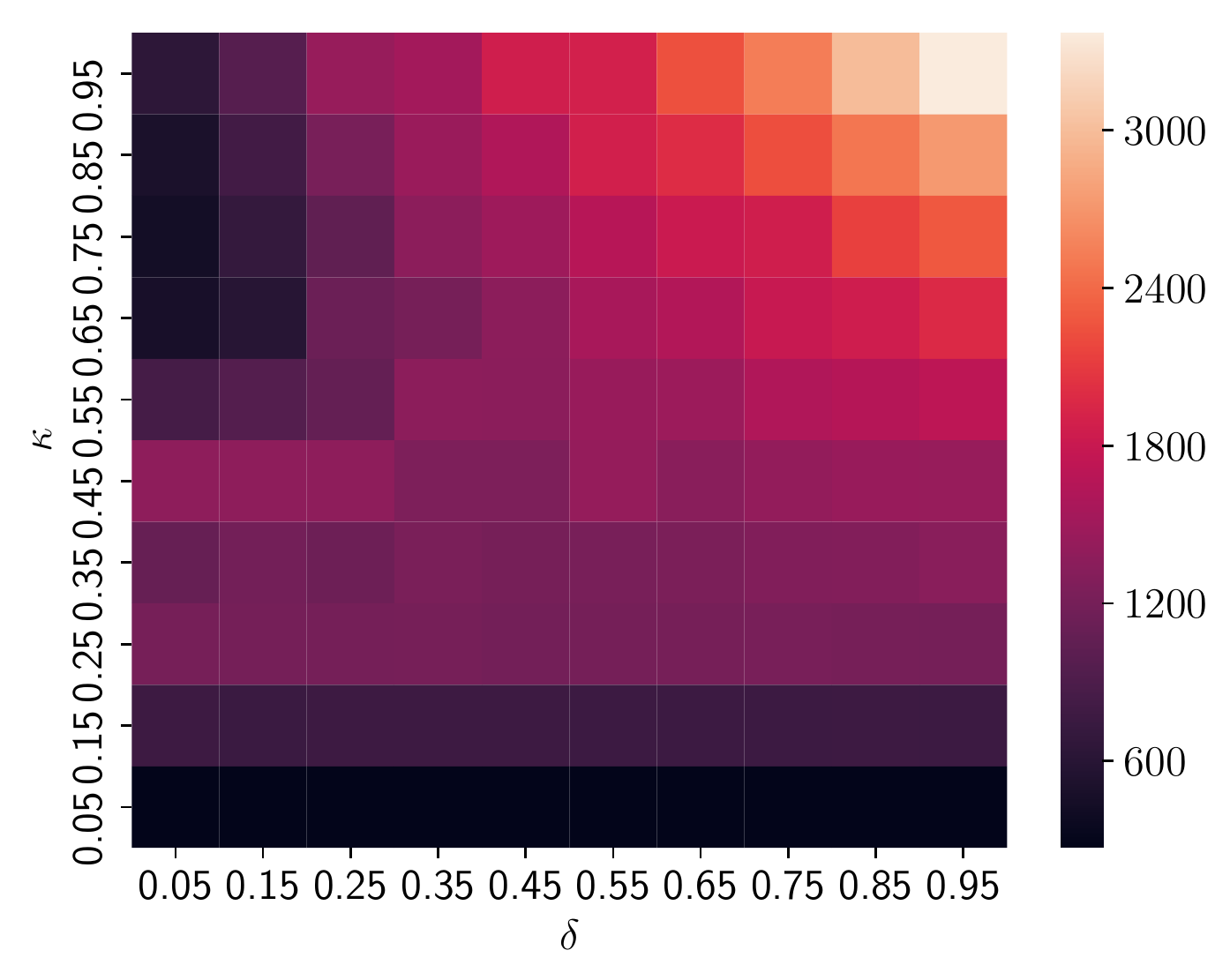}}
    \caption{\label{fig:kappa-delta}Performance and running time of FastGreedy as $\kappa$ and $\delta$ are simultaneously varied.}
\end{figure}

We further notice some interesting patterns in these heat-maps. When
$\kappa$ is near 0, the selection of $\delta$ does not appear to matter.
This is likely due to the role each parameter plays: $\kappa$ plays a
critical role in identifying a DR-violation, at which point $\delta$ is
the rate by which $\beta$ is reduced to compensate. When $\kappa$ is
small, it takes commensurately larger violations for $\delta$ to apply.
These larger violations are not seen in our simulations, and thus
$\delta$ has no impact when $\kappa$ is very small.
\section{Implementation Notes} \label{apx:implementation}
The \texttt{facebook} dataset originally is undirected; we replace each edge $u
\leftrightarrow v$ with two edges $u \rightarrow v$ and $v \rightarrow
u$. 

We calculated the FastGreedy ratio using the values $\kappa=0.95$,
$\beta^* = 0.9$, $\gamma_s = 0.69857$, $\epsi = 0$. The values of
$\kappa$ and $\epsi$ are taken from the parameters used to run the
algorithm ($\epsi$ can be 0 since the algorithm always 
returned $\vec g$ with $\lone{\vec g}=k$), 
and the values of $\beta^*$ and $\gamma_s$ are the minimum
over all instances where FastGreedy Submodularity Ratio 
was computed.

\subsection{Evaluating $\delta_\vec{s}(\vec{g})$}
As mentioned in Section \ref{sect:experiments}, we evaluate the objective
on a set of 10,000 Monte Carlo samples. However, for performance reasons
we do not evaluate the marginal gain $\delta_\vec{s}(\vec{g})$ directly by computing
each of $f(\vec{g} + \vec{s})$ and $f(\vec{g})$ and then subtracting. Instead, we compute
the expected number of activations across the sample set were $\vec{s}$ to be
added. We accomplish this as follows.

First, we mantain a state associated with the vector $\vec{g}$. This state
contains a number of variables for logging purposes, in addition to two
of note for our discussion here: \texttt{samples} and \texttt{active}.
The former is the list of sampled graphs, each represented as a pair of
vectors of floating point numbers corresponding to random thresholds
assigned to each node and edge. The latter is a list of sets of nodes
currently active in each sampled graph under solution vector $\vec{g}$. These
active sets are computed only when a new element (or elements) are added
to the solution vector and are computed directly by (a) computing the
set of externally activated nodes by checking if the random threshold is
sufficient to activate the node; then (b) propagating across any active
edges according to their thresholds. The code for this is contained in
the \texttt{active\_nodes} function of \texttt{src/bin/inf.rs} in the
code distribution.

Given this representation, to estimate the marginal gain of $\ell$ copies of a node $s$ \textit{on each sample} $S_i$ as follows:
\begin{enumerate}
    \item Check if the node is already active on $S_i$. If so, return 0 for this sample.
    \item Check if the node would be externally activated if added to the solution. If not, return 0.
    \item Check if the node would be activated by neighboring nodes if added to the solution. If not, return 0.
    \item Compute the set of nodes that would be newly activated if $s$ becomes active via breadth-first-search from $s$. Return this number $c$.
\end{enumerate}
Each would-be-activated check is accomplished by comparing the
activation probability of the node or edge to the random threshold
associated with it in the sample. Then the expected marginal gain is the
average result across all samples $S_i$. When an element (with
multiplicity) is chosen to be inserted into the
solution, we add it to the solution vector, discard and recompute all
samples, and then recompute the \texttt{active} set on each sample from
scratch. This is in line with prior Monte-Carlo-based solutions
for IM \cite{Kempe2003}.

While in theory it is possible to consider the marginal gain of an
arbitrary vector $\vec{v}$, in our implementation we restrict the values
it can take to $\vec{v} = \ell \vec{s}$, where $\vec{s}$ is the unit
vector for node $s$. This simplifies each of the steps above. The
implementation of the above is contained in the functions \texttt{delta}
and \texttt{scaled\_delta} in \texttt{src/bin/inf.rs}.

\end{document}